\documentclass[preprint,12pt]{elsarticle}
\usepackage[T1]{fontenc}
\usepackage{lmodern}
\input{glyphtounicode}
\usepackage{microtype}
\usepackage{amsmath,amssymb}
\usepackage{amsthm}
\usepackage{booktabs}
\usepackage{multirow}
\usepackage{hyperref}
\DeclareUrlCommand\doi{\urlstyle{rm}}  
\usepackage{xurl}
\graphicspath{{figures/}}

\newtheorem{theorem}{Theorem}
\newtheorem{proposition}{Proposition}
\newtheorem{corollary}{Corollary}
\newtheorem{definition}{Definition}

\makeatletter\@ifundefined{bibsep}{\newlength{\bibsep}}{}\makeatother
\makeatletter
\def\@listi{\leftmargin\leftmargini \topsep 3\p@ \parsep 1\p@ \itemsep 2\p@}
\let\@listI\@listi \@listi
\makeatother

\journal{Engineering Applications of Artificial Intelligence}

\begin{document}

\begin{frontmatter}

\title{Seed-Anchored Budget-Bounded Graph Rendering for Question Answering on Industry-Standard Power-Grid Information and Exchange Models}

\author[epri]{Jayakumar Manoharan\corref{cor1}\fnref{orcid1}}
\ead{jmanoharan@epri.com}

\author[epri]{Yamini Sehgal\fnref{orcid2}}
\ead{ysehgal@epri.com}

\affiliation[epri]{organization={Electric Power Research Institute (EPRI)},
            addressline={1300 W. T. Harris Blvd},
            city={Charlotte},
            state={NC},
            postcode={28262},
            country={USA}}

\cortext[cor1]{Corresponding author.}
\fntext[orcid1]{ORCID: 0009-0009-7765-9165}
\fntext[orcid2]{ORCID: 0009-0004-2364-2310}

\begin{abstract}
Large language model question answering over power-grid models must respect a
fixed context budget. The artificial-intelligence contribution is a render-boundary
locality prior with a precondition checkable before rendering: a deterministic,
seed-anchored rendering algorithm with no method-specific tuned or learned
parameters beyond the shared hop bound and context budget,
preserving every predefined seed-local answer-bearing render unit when
aggregate render mass $D(S)$ stays within budget $B$, holding for every
question evaluated. The mechanism is truncation, measured without a reader:
on both budget-binding networks naive rendering keeps evidence for every single-hop
item but only $0.12$ and $0.00$ of multi-hop items; seed-anchored keeps
$1.00$ on both strata. A join against reader correctness shows this
failure is truncation, not reader error. The engineering application is question
answering over Common Information Model (CIM) models under European Network of
Transmission System Operators for Electricity (ENTSO-E) Common Grid Model Exchange Standard
(CGMES) conformity validation. On a fresh SmallGrid-family bank, the pre-registered
primary endpoint confirms: accuracy rises from $0.450$ to
$0.970$ ($n{=}100$, one-sided McNemar $p{=}2.2\times10^{-14}$), item-level
confirmation within one topology family, not network-level generalization. On an independent
$72{,}418$-object topology the attribute-lookup stratum shows no harm, all arms at $1.000$;
its multi-hop stratum is retired (non-unique gold referents). Under a common pipeline,
the seed-anchored standards-native arm matched or exceeded three GraphRAG frameworks' extracted
graphs at zero construction tokens, across MicroGrid, SmallGrid and SmallGrid-3.0.
The endpoint is specific to one 7B reader at an 8{,}000-character budget; the evidence
concerns budget-bounded retrieval, not general question answering.
\end{abstract}

\begin{keyword}
retrieval \sep graph \sep budget \sep question answering \sep power grid \sep CIM
\end{keyword}


\end{frontmatter}

\section{Introduction}
\label{sec:intro}

Developing scalable retrieval methods for large-language-model (LLM) assistants is a central
challenge in artificial intelligence for engineering. When the knowledge source consists of detailed
power-system representations encoded in the Common Information Model (CIM)~\citep{iec61970} and
exchanged via the Common Grid Model Exchange Standard (CGMES)~\citep{entsoe2014cgmes}, the retrieval
system must select a subgraph that fits within the finite token context of an LLM while preserving
all answer-relevant information. This paper treats that problem as a budget-bounded graph-rendering
task and evaluates it on European Network of Transmission System Operators for Electricity
(ENTSO-E) CGMES conformity-assessment configurations and an independently sourced larger grid
topology; results do not establish performance on deployed operational models. The engineering
motivation is grounded lookup over a utility's own grid model rather than open-domain question
answering, with intended applications such as equipment-attribute retrieval for planning and
protection review, topology and containment lookup during model validation and ENTSO-E conformity
assessment, operator exploration of an unfamiliar export, and comparison across model revisions.
If a retrieval failure occurs under a fixed context budget, it could undermine the reliability of
automated grid audits, planning tools, and operational decision support, so an assistant's cost,
data-governance, and fidelity matter as much as its raw accuracy. Four evaluation tiers are
distinguished throughout: the six ENTSO-E CGMES conformity-test configurations; the RealGrid
topology; synthetic scaling networks used for answer-presence evaluation only; and operational
transmission systems, which are not evaluated here.

The default recipe for such assistants is graph-based retrieval-augmented generation
(GraphRAG), which specialises retrieval-augmented generation~\citep{lewis2020_rag} to a graph:
build a knowledge graph over the content, retrieve a query-relevant subgraph, and let an LLM reader
answer from the rendered context~\citep{edge2024_from_local_global, guo2024_lightrag_simple_fast}.
Almost every recent power-domain GraphRAG system builds that graph by prompting an LLM to extract
entities and relations from text~\citep{2026_hgrag_hierarchical_graphenhanced}, unlike earlier
hand-engineered power knowledge graphs~\citep{tang2021_intelligent_question_answering}. That choice
is expensive and lossy: LLM extraction spans 0.03M to 1.7M tokens per network across the three
frameworks we test (31{,}890 to 206{,}254 for HippoRAG; 87{,}962 to 754{,}781 for LightRAG; 235{,}312 to 1{,}703{,}707 for Microsoft GraphRAG), recurs on every model revision, and recovers only part of the source relational
structure~\citep{chen2025_are_large_language, mihindukulasooriya2023text2kgbench}, so numeric
attributes migrate between similarly named components. Yet power systems already ship a curated,
schema-typed relational model in the CGMES export, so we ask whether this standards-native graph
can be read directly under a fixed context budget rather than reconstructed by an LLM. A prior
controlled study~\citep{prior2026_trustworthy_kg_cgmes} reported the puzzle we diagnose here:
parsed directly from CGMES, the standards-native graph was strictly more faithful yet lost badly on
the larger network under a fixed 8{,}000-character budget.

\textbf{The failure is retrieval, not representation.} We cast budget-bounded retrieval as an
ordering over render units (entity descriptions and edge renderings) truncated at a budget $B$.
The common naive order, which renders every entity description before any edge, must fail whenever
a high-degree hub (a shared voltage level, region container, or busbar; degree 381 to 481 in our
data) floods the budget before the answer-bearing edges are reached (Theorem~\ref{thm:separation}).
The failure is representation-dependent and flips with the encoding: the raw CIM graph collapses on
the bus-branch SmallGrid yet survives on its node-breaker CGMES 3.0 counterpart, where the derived
support graph collapses instead (Section~\ref{sec:results-flip}).

\textbf{A budget-safe fix, and where its novelty lies.} We introduce seed-anchored retrieval, a
deterministic render order that introduces no additional tuned or learned parameters and that enforces the locality prior at the render/truncation
boundary through a checkable greedy-prefix precondition ($D(S)\le B$). The guarantee covers
the predefined seed-local answer-bearing render units whenever the seed-local render mass fits the
budget; we do not claim that hub-induced overflow is ruled out in general. Concretely, it
renders the query object's own description and its incident edges first, then fills outward by hop
distance. The retriever consumes zero LLM graph-construction tokens, introduces no learned parameters
and no tuned degree threshold, and uses the same hop bound $k$ and budget $B$ as the naive baseline.
Its methodological contribution is \emph{where} and \emph{how} the locality prior is enforced, at the
render/truncation boundary with a provable a priori boundary rather than at the retrieval-scoring
stage. In the evidence-scoped answer-presence audit, naive rendering places the evidence-linked
gold content in context for $0.56$ of questions on SmallGrid and $0.50$ on SmallGrid-3.0, and
seed-anchored rendering for $1.00$ on both; where the budget does not bind, both reach $1.00$.
Ego-graph expansion and hub-pruned BFS, rerun under the same test, also reach $1.00$ on both
budget-binding networks, so the contribution is not presence superiority over locality priors but
the checkable a priori guarantee and the fact that the method introduces no additional method-specific tuned or learned parameters beyond the shared hop bound and context budget. Here,
\emph{budget-binding} is an operational label with a measured criterion: naive serialization must
both reach the 8{,}000-character cap and actually evict evidence, the latter established by the
evidence-scoped presence audit falling below $1.00$. Cap hits alone do not qualify. MicroGrid
illustrates why the second condition is needed: naive hits the cap on 48 of 50 items there, yet its
evidence-scoped presence remains $1.00$, so the cap is reached without evidence being lost and
MicroGrid is not budget-binding. SmallGrid ($0.56$) and SmallGrid-3.0 ($0.50$) fall below $1.00$
and are budget-binding. The label does not mean that $D(S)>B$. Nor is the problem answered by ranking the budget fill by relevance: a
score-weighted fill, standard in token-budgeted context selection~\citep{peng2025_adagres_budgeted},
does not protect the answer, because a high-degree hub can itself be highly relevant to the query and
still flood the window and evict the seed-incident answer edges, whereas ordering the fill by hop
locality renders those edges first regardless of their relevance score. What seed-anchored adds is a
checkable greedy-prefix guarantee at that boundary: whenever the seed-local render mass fits the
budget, every seed-local answer is preserved (an elementary answer-preservation property, not a deep
theorem, Theorem~\ref{thm:preservation}) and the order is never worse than naive rendering in answer
presence (no-harm dominance, Proposition~\ref{prop:noharm}), the assurance a deployer needs before
swapping a renderer and the one our failed threshold heuristic lacked. It is not a new proximity
score but a deployment-safe rendering rule over the standards-native CIM/CGMES graph with no LLM
extraction. That failed first attempt, a degree-threshold heuristic tuned on the development network
that regressed on a new one, we report as validation history (Supplementary Material, Section~S1).

\textbf{Evaluation with a development and held-out split.} We evaluate on six ENTSO-E CGMES
conformity networks spanning 233 to 5{,}315 objects and two CGMES versions (2.4.15 and 3.0); the
CGMES 3.0 networks are the same underlying grids re-encoded in the newer version and are held out from
retrieval-method, seeding, hop-bound, threshold, and scorer development, with parser compatibility
fixes for nullable 3.0 fields disclosed separately, so the evaluation tests robustness to the encoding-version change. Across twelve arms (nine graph
arms, being two standards-native representations under naive rendering, the seed-anchored method,
and three LLM-extracted graphs from three frameworks read under both render orders; plus three
corpus-level baselines), five question strata, deterministic scoring backed by two LLM
judges, and pre-registered statistics, three findings emerge. Because the templated benchmark is
self-referential (83\% of gold answers appear verbatim), the strongest non-self-referential evidence
is the 100-question human-phrased stress test (Section~\ref{sec:results-human}, method $0.83$).

\begin{enumerate}
\item \textbf{Seed-anchored retrieval realizes its guarantee, and the mechanism it acts on is
truncation.} Measured without a reader and split by stratum, naive rendering retains the
evidence-linked gold content for every single-hop item on both budget-binding networks, but for
only $0.12$ of multi-hop items on SmallGrid and none on SmallGrid-3.0, while seed-anchored
rendering retains all of them on both strata. An item-level join against reader correctness shows
the multi-hop failure is truncation rather than reader error: where the evidence survived, the
reader used it. This mechanism, not any single accuracy delta, is the paper's load-bearing result.
Seed-anchored retrieval was not observed to underperform
naive on any network. After Holm correction, accuracy improves significantly on three networks:
MicroGrid ($0.68$ to $0.92$, $p{=}7.3\times10^{-3}$) and the two budget-binding networks,
SmallGrid ($0.52$ to $0.98$, $p{=}1.2\times10^{-6}$) and SmallGrid-3.0 ($0.42$ to $0.96$,
$p{=}8.9\times10^{-8}$). The two budget-binding networks clear the conservative minimum detectable
effect of $0.280$; MicroGrid is Holm-significant but falls below it, and is not budget-binding.
The other three per-network comparisons are underpowered
and are treated as inconclusive, resting on a retrieval-level no-harm property under the stated
precondition rather than on a statistical non-inferiority claim about end-to-end accuracy. On the
larger independent RealGrid topology, whose multi-hop stratum is retired for want of unique gold
referents, the surviving attribute-lookup stratum shows all three deterministic arms at $1.000$, a
no-harm result rather than a dominance one (Section~\ref{sec:results-main}). The pre-registered
primary confirmatory endpoint is a pooled fresh-bank test drawn on the two SmallGrid encodings
(Section~\ref{sec:results-main}); it is strong item-level confirmation of the mechanism within that
one topology family and is not a network-level generalization claim. The human-authored probes and
RealGrid are reported separately as robustness and transfer evidence and are not folded into it.
\item \textbf{Across three GraphRAG frameworks (LightRAG, Microsoft GraphRAG, and HippoRAG), under
the tested extractor, reader, and budget, standards-native seed-anchored retrieval matches or exceeds
the extracted graph representations those frameworks produce, read under a common
retrieval-and-rendering pipeline, at zero graph-construction tokens\linebreak
(extractor-conditional)}
(extraction cost 0.03M to 1.7M tokens per network, which excludes build, inference, and deployment
cost; Section~\ref{sec:results-gllm}). One framework's native retrieval scheme, HippoRAG's
personalized PageRank, was reimplemented and evaluated; the packaged native retrievers of all
three frameworks were not run. The
comparison is extractor- and reader-conditional: the
development-network parity held only against a frontier proprietary extractor, and swapping in an
open-weight one drops LLM-extracted GraphRAG to 0.66 while anchored stays 0.98
(Section~\ref{sec:results-extractor}). At fixed anchored order the standards-native graph still leads
on the held-out network where the budget binds (0.96 versus 0.80) and on MicroGrid, where it does
not (0.92 versus 0.82), a graph-source and
fidelity effect, not a rendering-order artifact.
\item \textbf{The finding generalizes beyond graphs.} Simple lexical retrieval over the
deterministic corpus rendering is also strong (0.86 to 0.98 on single-hop and topology strata),
with the caveat that these questions are template-generated from that corpus: what matters is a
faithful standards-native representation plus budget-aware retrieval, not an LLM in the
graph-construction loop.
\end{enumerate}

\textbf{Contributions.} Our load-bearing contribution is empirical. (i) A controlled, cross-version,
twelve-arm evaluation on six ENTSO-E CGMES conformity configurations, three of which carry
LLM-extracted graphs and so support the extraction comparison. At zero LLM graph-construction
tokens, and under the reader, extractor, budget and framework set stated in
Section~\ref{sec:method}, the standards-native graph matches or exceeds extracted graphs from three
extraction paradigms (LightRAG, Microsoft GraphRAG, HippoRAG; 0.03M to 1.7M tokens per network)
read under a common pipeline, with one native retrieval scheme reimplemented as a check. At fixed
rendering order the standards-native graph still leads, because the tested extractors lose
relational and numeric fidelity. This ties to the energy-domain finding that which CGMES encoding
overflows the budget flips between standard versions, so a deployer cannot predict it
(Section~\ref{sec:results-flip}). (ii) The seed-anchored retriever, a deterministic render order
that introduces no additional tuned or learned parameters and enforces the
locality prior at the render/truncation boundary through a checkable precondition
($D(S)\le B$) rather than at the retrieval-scoring stage, so that, unlike a
relevance-ranked budget fill, a hub that is itself relevant cannot evict the seed-incident answer
edges. (iii) A
formal characterization that makes the truncation failure and the no-harm condition explicit and
checkable: an elementary greedy-prefix answer-preservation property (Theorem~\ref{thm:preservation}),
its no-harm dominance corollary (Proposition~\ref{prop:noharm}), and a quantitative naive-failure
characterization (Theorem~\ref{thm:separation}, Proposition~\ref{prop:quant}).
(iv) Evidence-bounded practical guidance, including the cost and data-governance case for reading
the standard in place rather than shipping it to an external extractor, with a model-free scaling
demonstration to $10^5$ objects (validating answer-presence, not end-to-end reader QA accuracy) and a
100-question human-phrased robustness probe that bound the envelope of support. (v) On a second,
independently sourced base topology, RealGrid, the no-harm property holds on the attribute-lookup
stratum, where all three deterministic arms reach $1.000$. That topology's multi-hop stratum is
retired and reported as retired: its gold labels do not identify unique referents, so it cannot
support a topology-retrieval claim in either direction. In
the power-system and CGMES setting we are, to our knowledge, first to analyze budget-bounded
retrieval directly on
ENTSO-E CGMES test configurations and first to tie a standards-native-versus-LLM-extracted
graph-source result to a budget-preservation guarantee; the underlying ordering and extraction-free
themes have concurrent analogues in other domains (Section~\ref{sec:related}).

\textbf{Generality beyond the power-grid application.} The seed-anchored render-boundary rule is a
generic budget-bounded retrieval-augmented-\linebreak
generation preprocessing lever, not a
power-grid-specific device. It reorders the retrieved subgraph so that the query object's
description and its incident edges are rendered first and further hops are added only while the
cumulative render mass stays within the budget, which makes it deterministic, free of tuned or
learned parameters, and equipped with a provable greedy-prefix guarantee that preserves the
seed-local answer-bearing render units whenever their render mass fits the budget. The guarantee is
domain-agnostic in its proof, which assumes only a graph, a render vocabulary and a budget. It is
not domain-agnostic in its evidence: it is validated here only on CIM and CGMES topologies, and we
make no empirical claim about other high-degree-hub retrieval settings. What is domain-specific is
the render vocabulary and the seeding rule, not the statement of the guarantee. We validate the rule on power-grid models expressed in CIM and exchanged via
CGMES, where high-degree hubs and a hard context budget make the failure mode acute and a
standards-native graph is available without an extraction pass; that engineering setting is the
validating application, not the limit of applicability. Positioned this way, the contribution in
artificial intelligence, a render-boundary locality prior with an a priori answer-preservation
guarantee, leads, and the CGMES question-answering study is its validation.

\section{Related work}
\label{sec:related}

\textbf{GraphRAG and when it helps.} Letting a language model build a graph and summarize its
communities for global sensemaking was popularized by Microsoft's
GraphRAG~\citep{edge2024_from_local_global}; follow-on systems pursued cheaper construction
(LightRAG~\citep{guo2024_lightrag_simple_fast}) or a memory index
(HippoRAG~\citep{gutirrez2024_hipporag_neurobiologically_inspired}). As these pipelines were
benchmarked more carefully, a qualified picture emerged: on ordinary factoid questions graph-enhanced
retrieval is often no better than plain retrieval~\citep{xiang2025_when_use_graphs,
chen2025_comparing_rag_graphrag}, and its advantage erodes once an agent can
search~\citep{fan2026_still_need_graphrag}; the design space has been
surveyed~\citep{han2024_retrievalaugmented_generation_graphs}, with gains concentrated in
relation-dense engineering settings~\citep{ahmad2025_benchmarking_vector_graph,
electronics2025_document_graphrag_knowledge}. We contribute a mechanism-level account of one recurrent
failure: once a context budget is fixed, the deciding factor is often not the graph but the order in
which units render up to the budget boundary.

\textbf{Where a locality prior is enforced.} A locality prior is central to seed-anchored retrieval,
and such priors are familiar in graph retrieval: HippoRAG ranks passages by personalized PageRank
over a graph seeded at the query entities~\citep{gutirrez2024_hipporag_neurobiologically_inspired},
and ego-graph or proximity-expansion retrievers score candidate nodes and edges by graph distance
from the seeds. Our point of departure is the \emph{stage} at which the prior takes effect. Those
methods act at the retrieval-scoring boundary, settling \emph{which} units enter a candidate set;
but once that set is laid into a bounded context, a second and usually unstated ordering settles
\emph{which units survive truncation}, and it was there, not at scoring, that the standards-native
graph failed in the prior study. Seed-anchored retrieval places the prior at this
render/truncation boundary, so that with a total render order and greedy budget fill it becomes a
\emph{checkable} precondition: whenever the seed-local render mass fits the budget,
$D(S)\le B$, answer presence is guaranteed rather than merely made likely
(Section~\ref{sec:theory}). A proximity-scored retriever that does not govern render order carries
no comparable guarantee, since a hub that clears the scoring stage can still flood the window and
evict the seed-incident answer edges. The same reasoning rules out simply ranking the budget fill by
relevance, as in score-weighted truncation and token-budgeted context
selection~\citep{peng2025_adagres_budgeted}: a high-degree hub can itself be highly relevant to the
query and still flood the window, so a relevance-ranked fill does not protect the seed-incident answer
edges, whereas a hop-locality fill renders them first regardless of their relevance score. Consistent
with this, these priors and seed-anchored perform alike at reader-free answer presence; what is new is
treating the render/truncation order itself as the object of a locality prior with an a priori
answer-preservation precondition, layered on top of, not competing with, proximity-based retrieval.

\textbf{GraphRAG and LLM evaluation in the power domain.} Recent work carried these ideas into power
systems: HG-RAG constructs a hierarchical, LLM-extracted graph for power-system question
answering~\citep{2026_hgrag_hierarchical_graphenhanced}, GridCodex applies a RAG-driven framework to grid
codes~\citep{shi2025_gridcodex_ragdriven_framework}, and pre-LLM systems built knowledge-graph QA for
the power domain~\citep{tang2021_intelligent_question_answering}; on the evaluation side, ElecBench
probes LLMs on power dispatch~\citep{zhou2024_elecbench_power_dispatch} and industry assessments
caution about LLM-assistant reliability in operational settings~\citep{epri2025_epri_2025_llm}. None
treat the engineering standard itself as the graph source and ground truth; closing that gap was the
aim of a prior controlled study~\citep{prior2026_trustworthy_kg_cgmes}, which the present paper
extends with a formal characterization, a new retrieval method, and a six-network cross-version
evaluation.

\textbf{CIM/CGMES as a semantic resource.} Viewing the CIM as an ontology is long-standing
practice~\citep{gaha2013ontology, schumilin2017ontology}; IEC 61970-501 fixes the CIM RDF
schema~\citep{iec61970_501}, and ENTSO-E publishes the conformity test configurations we
use~\citep{entsoe2014cgmes}. Because that schema makes a CGMES export directly queryable, this line
has long treated the standard as a semantic graph accessed via SPARQL rather than reconstructed by a
model, so reading the standard directly is itself established practice; our contribution is not that
move but tying it to budget-bounded rendering with a preservation guarantee. A parallel line of work
measures how faithfully LLMs reconstruct structured knowledge from
text~\citep{mihindukulasooriya2023text2kgbench, chen2025_are_large_language}, and concurrent GraphRAG
work in other domains likewise argues that LLM extraction is avoidable, building the index from
data-derived structure for generic documents~\citep{prosvirnin2026_contextrag_extractionfree} or from
deterministic source-code structure~\citep{chinthareddy2026_ast_vs_llm_graphrag}. Our findings offer
the complementary case in which a canonical structure is already present as a governed engineering
standard: the right move is not to extract it more faithfully but not to extract it at all.

\textbf{CGMES tooling and conformity assessment.} Established power-system tooling for CGMES, the
network-analysis libraries powsybl~\citep{powsybl} and pandapower~\citep{thurner2018pandapower} and
the ENTSO-E conformity-assessment workflow, assumes full in-memory access to the parsed model or
operates on small inputs, and does not ask what survives when that model must be rendered into a
bounded LLM context. In the power-system and CGMES setting, and to our knowledge, this is the first
work to analyze budget-bounded retrieval directly on ENTSO-E CGMES
configurations and to tie that analysis to a budget-preservation guarantee, positioning the
contribution against power-system graph retrieval and not only generic AI/RAG pipelines.

\textbf{Rendering into a context budget.} Any deployable GraphRAG system must serialize retrieved
structure into a bounded context window, and overstuffing that window is known to degrade access to
relevant content placed in the middle of a long context~\citep{liu2023_lost_in_the_middle}; the
specific failure of a high-degree entity whose neighbourhood exceeds the window and is involuntarily
truncated also surfaces in graph-retrieval attack
settings~\citep{gu2026_graphsteal_supernode}. Our contribution is not to discover this phenomenon but
to tie it quantitatively to a specific naive render order (Section~\ref{sec:theory}). Budget-aware
selection itself is well studied: token-budgeted context selection optimizes a relevance and
redundancy objective under a budget~\citep{peng2025_adagres_budgeted}. Our answer-preservation statement is weaker and
narrower than that machinery, an exact but elementary greedy-prefix property that holds only for the
locality-priority order and only for answer presence, not an approximate optimality bound over a
relevance objective. What we add is placing that checkable precondition at the render/truncation
boundary for a standards-native graph; we do not claim a general new guarantee for GraphRAG context
assembly.

\section{Background}
\label{sec:background}

\textbf{CGMES models.} A CGMES export is a collection of RDF/XML profiles (equipment, topology,
steady-state hypothesis, state variables) describing a network as typed objects (substations, voltage
levels, lines, transformers, machines, topological nodes) joined by schema-defined
associations~\citep{iec61970, iec61970_501, entsoe2014cgmes}. Because the schema is curated and the
associations explicit, the export is already a relational model of the grid before any learning
component is introduced. We parse it deterministically into three artefacts.
The first is a standards-native graph $G_{\mathrm{STD}}$ (raw CIM graph), in which every object
becomes a node carrying its numeric attributes and every retained schema association becomes an
edge. The second is a derived support graph $G_{\mathrm{HYB}}$, in which the raw terminal-mediated
associations are replaced by five derived relations: containment, electrical connectivity,
node-connection, voltage, and part-of. The third is a deterministic corpus, a template rendering of the same facts into text (for example,
``AC line segment 82-96 connects topological node Baileysv to topological node Logan. It has a length
of 101.782 km, a resistance of 2.82269 ohm, a reactance of 9.23472 ohm, and operates at 132 kV.'').
No language model is involved in producing the ground truth, the corpus, or the questions.

\textbf{The frozen QA pipeline.} Every arm shares a single pipeline, following the prior controlled
study~\citep{prior2026_trustworthy_kg_cgmes}. Seed entities are located by matching node names
against each question; the retrieved structure is rendered into text and hard-capped at $B = 8{,}000$
characters (stated in characters because the renderer is deterministic and tokenizer-independent); a
frozen local reader, Qwen2.5-7B-Instruct in bfloat16 with greedy decoding~\citep{qwen2.5}, produces
an answer, and a deterministic scorer checks it against a machine-checkable gold specification. Across
arms only the underlying graph and its render order differ; the seeder, budget, reader, and scorer are
fixed. Questions fall into five strata: S1 single-hop attribute lookups (for example, for a given
transformer, what is its nominal voltage?), S2 seed-local line-to-node topology (for example, which
lines connect one bus to another under the CGMES 3.0 node-breaker representation?), S3 multi-hop
containment chains, S4 aggregation counts, and S5 global superlatives.

\textbf{The problem.} On the larger development network the prior study reported $G_{\mathrm{STD}}$
under the standard rendering at 0.48 overall against 0.98 for a language-model-extracted graph, even
though the standards-native graph contained strictly more correct structure; the proximate signal was
truncation, with 45 of 50 contexts reaching the cap. This raises the central question: is the deficit
intrinsic to standards-native representations, or an artefact of rendering into a bounded context, and
if the latter, can a rendering order with a provable guarantee close the gap at zero language-model
graph-construction tokens?

\section{Budget-bounded retrieval on CGMES power-system graphs}
\label{sec:theory}

This section is a formal characterization of the empirical mechanism of
Section~\ref{sec:results}, not a body of deep theory. The paper's load-bearing contribution is
empirical, the zero-token cost and fidelity-versus-extraction results of
Section~\ref{sec:results-gllm} and the encoding-dependent overflow flip of
Section~\ref{sec:results-flip}; the role of this section is narrower, to make the truncation failure
and the no-harm condition explicit and checkable. The statements are intentionally simple, elementary
greedy-prefix facts about render order under truncation rather than new mathematics. The one a
practitioner needs is no-harm dominance (Proposition~\ref{prop:noharm}): seed-anchored rendering is
never worse than naive rendering in answer presence, so within the guaranteed regime (seed-local
answers whose render mass fits the budget) adopting it cannot regress a query, the
assurance needed before replacing a deployed renderer and the one our failed threshold heuristic
(Supplementary Material, Section~S1) lacked. Every other statement supports it: answer
preservation (Theorem~\ref{thm:preservation}) and the naive-rendering failure mode it improves upon
(Theorem~\ref{thm:separation}, Proposition~\ref{prop:quant}).

Let $G=(V,E)$ be the retrieval graph (undirected; each edge carries a relation label), and
fix a context budget $B$ in characters (the frozen pipeline uses $B{=}8{,}000$).

\begin{definition}[Render units and orderings]
Each node-description occurrence $x=\mathrm{desc}(v)$ and each edge-rendering occurrence
$x=\mathrm{rend}(e)$, for $e=(u,r,w)$, is a render unit identified by its source node or edge;
occurrences remain distinct even if their rendered strings coincide. The quantity
$\mathrm{len}(x)$ is the unit's rendered length including its separator. Fix injective,
deterministic identifiers $\mathrm{id}_V:V\to\mathbb{N}$ and
$\mathrm{id}_E:E\to\mathbb{N}$, and set
\[
\mathrm{id}(\mathrm{desc}(v))=\mathrm{id}_V(v),
\qquad
\mathrm{id}(\mathrm{rend}(e))=\mathrm{id}_E(e).
\]
For a query $q$ with seed set $S=S(q)\subseteq V$ (the deterministic name-match rule shared
by all arms) and hop bound $k$ (the frozen experiments use $k=2$), the candidate universe is
\begin{multline*}
U_k(S)=\{\mathrm{desc}(v): d(v,S)\le k\}\\
\cup\{\mathrm{rend}(e): e=(u,r,w),\ d(u,S)\le k,\ d(w,S)\le k\},
\end{multline*}
where $d(\cdot,S)$ is hop distance to the nearest seed. The hop level of a unit is
\[
\mathrm{hop}(\mathrm{desc}(v),S)=d(v,S),\qquad
\mathrm{hop}(\mathrm{rend}(e),S)=\max\{d(u,S),d(w,S)\}.
\]
Thus $U_k(S)$ is exactly the set of candidate units with hop level at most $k$. Define the
anchored tier
\[
\mathrm{tier}_S(x)=
\begin{cases}
0, & x=\mathrm{desc}(v),\ v\in S,\\
1, & x=\mathrm{rend}(e),\ e\text{ is incident to a vertex in }S,\\
2d(v,S), & x=\mathrm{desc}(v),\ v\notin S,\\
2\,\mathrm{hop}(x,S)+1, &
x=\mathrm{rend}(e),\ e\text{ is not incident to }S.
\end{cases}
\]
The naive and anchored rank keys are
\[
\begin{aligned}
\kappa_{\mathrm{nai}}(x)
&=
\begin{cases}
(0,\mathrm{id}(x)), & x\text{ is a description unit},\\
(1,\mathrm{id}(x)), & x\text{ is an edge-rendering unit},
\end{cases}\\
\kappa_{\mathrm{anc}}(x)
&=(\mathrm{tier}_S(x),\mathrm{id}(x)).
\end{aligned}
\]
The orders $\pi_{\mathrm{nai}}$ and $\pi_{\mathrm{anc}}$ are the increasing lexicographic
orders of these keys. Injectivity of the identifiers makes both keys total orders on every
finite $U_k(S)$. The context $C(\pi,B)$ is the set of units in the maximal prefix of $\pi$
whose total length is at most $B$ (the greedy hard cap used by the pipeline).
\end{definition}

\textbf{Standing assumptions.} $G$ is finite, so $U_k(S)$ is finite; every render unit has
positive length; and $\mathrm{len}$, measured in characters of the rendered UTF-8 text,
includes the fixed single-newline separator, so a prefix's length is the sum of its units'
lengths. All hold for the CGMES render units used here.

\begin{definition}[Answer presence and seed-locality]
The gold answer of $q$ is verifiable from an answer-bearing unit set (minimal with respect to
set inclusion) $A(q)\subseteq U_k(S)$; for S1 attribute questions
$A(q)=\{\mathrm{desc}(s)\}$ for a seed $s$, and for S2 topology questions $A(q)$ is the set
of answer-bearing edge-rendering units incident to the seed. Query $q$ is
\emph{retrieval-satisfied} by $\pi$ iff $A(q)\subseteq C(\pi,B)$. The set $A(q)$ is
\emph{seed-local} iff every unit in it is a seed description or a seed-incident edge
rendering. Let
\[
\begin{aligned}
\mathrm{SL}(S)=\{x\in U_k(S):{}&
x\text{ is a seed description}\\
&\text{or a seed-incident edge rendering}\}
\end{aligned}
\]
and $D(S)=\sum_{x\in\mathrm{SL}(S)}\mathrm{len}(x)$.
\end{definition}

The analysis compares the two total orders above (Figure~\ref{fig:ordering}).
The naive order renders all entity descriptions by node identifier and then all edges by edge
identifier. The anchored order renders seed descriptions first, then every seed-incident edge,
and then, for $j=1,\ldots,k$, the as-yet-unemitted $j$-hop descriptions followed by the
as-yet-unemitted $j$-hop edges, with identifiers breaking ties inside each tier. In particular, the
seed-local units are the first two tiers; a seed-incident edge is not emitted again in its
ordinary hop tier.

\paragraph{Scope: the analysed stream and the implemented stream}
The analysis above assigns an edge a hop level using the maximum endpoint distance, that is
$\max\{d(u,S),d(w,S)\}$. The implementation instead assigns edge priority using the minimum
endpoint distance, $\min\{d(u,S),d(w,S)\}$. The two orders therefore agree on seed-local units but
can differ beyond that tier. In particular, under the implemented min-endpoint priority, an edge
whose endpoint distances are $(1,2)$ is placed in the same edge-priority tier as an edge whose
endpoint distances are $(1,1)$, even though their max-hop levels are $2$ and $1$ respectively.
Identifier tie-breaking can therefore place an out-of-radius $(1,2)$ edge before a remaining
radius-1 $(1,1)$ edge, so the set of max-hop radius-1 units need not form a prefix of the
implemented order. Consequently, the $r$-local greedy-prefix argument does not apply to the shipped
min-hop renderer. Theorem~\ref{thm:preservation} and Proposition~\ref{prop:noharm} are unaffected,
because they require only that seed descriptions and seed-incident edges precede non-seed-local
units, which both orders satisfy. We therefore make no $r$-local preservation claim beyond
seed-local evidence for the evaluated implementation; the max-hop analytical extension is reported
separately in Supplementary Material, Section~S2. The empirical results are unaffected: every
reported number comes from the implemented order.

\begin{figure}[t]
\centering
\includegraphics[width=\linewidth]{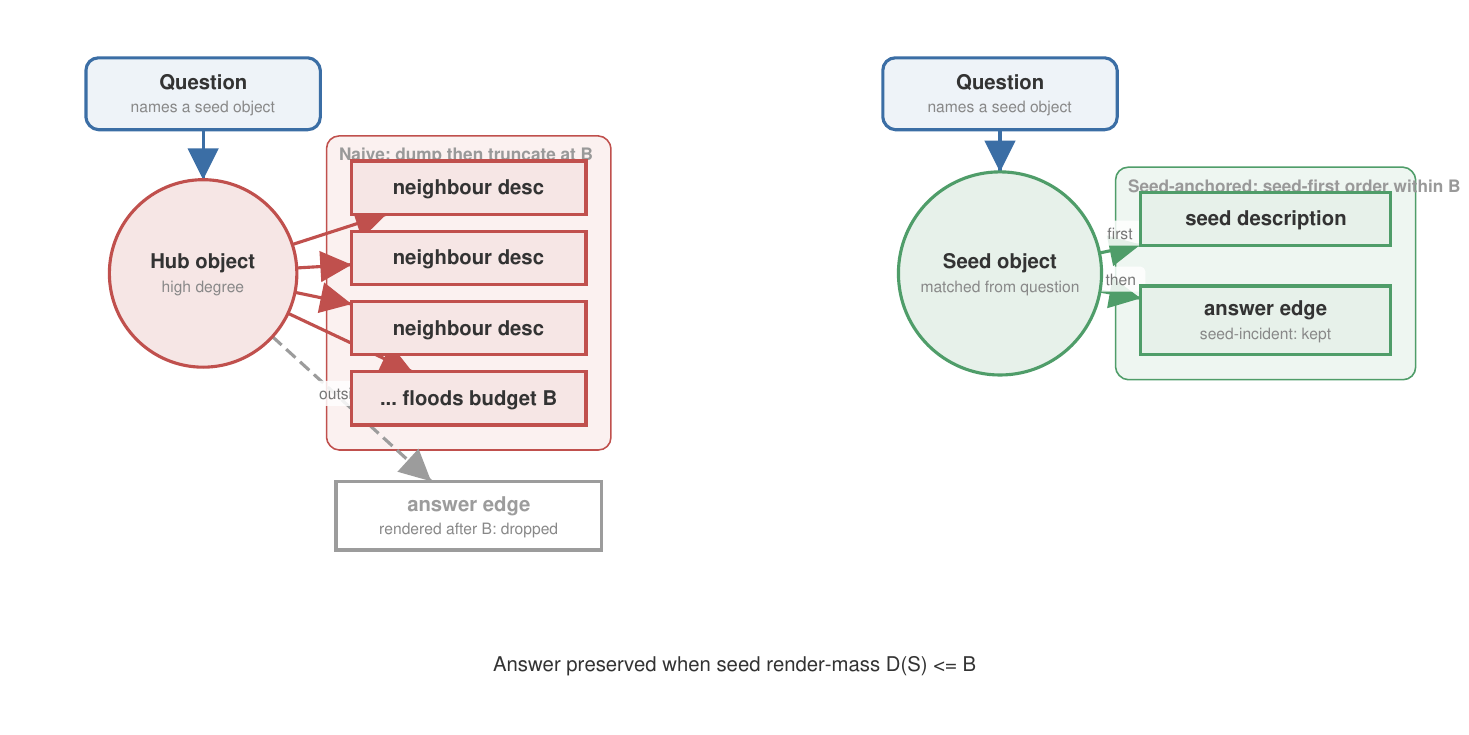}
\caption{Rendering order under a fixed budget. Naive dump-then-truncate emits all
descriptions before any edge, so a high-degree hub floods the budget and truncates the
seed-incident answer edges. Seed-anchored rendering emits the seed and its incident edges
first, then fills outward by hop, preserving budget-fitting seed-local answers.}
\label{fig:ordering}
\end{figure}

\begin{theorem}[Answer preservation]
\label{thm:preservation}
If $A(q)$ is seed-local and $D(S)\le B$, then $A(q)\subseteq C(\pi_{\mathrm{anc}},B)$.
\end{theorem}
\begin{proof}
$\pi_{\mathrm{anc}}$ renders every unit of $\mathrm{SL}(S)$ strictly before any
non-seed-local unit. Their total length is $D(S)\le B$, so the greedy fill admits all of
$\mathrm{SL}(S)$ into the budget prefix; $A(q)\subseteq \mathrm{SL}(S)$ concludes.
\end{proof}

In power-system terms, seed-anchored ordering prioritises the seed-local neighbourhood, the
equipment a queried line or transformer actually connects to, so localized energy-system queries about
that component stay safe under the budget.
The precondition $D(S)\le B$ is checkable a priori by summing the seed and seed-incident edge lengths
before retrieval (cost $O(|S|\cdot\deg)$), and it is mild: seeds are low-degree query entities (a
line, a transformer), and $D(S)\ll B$ for all S1/S2 items in the studied networks. The anchored order is not globally nondecreasing in hop when there are multiple seeds:
its seed-incident tier may contain both hop-0 and hop-1 edges. Nevertheless, every
seed-incident edge has hop at most 1, and after that tier the un-emitted units are ordered
outward by hop.

\paragraph{Analytical extension, and why it is not a guarantee of the evaluated method}
Under a max-hop edge-priority variant, all render units within radius $r$ form a prefix whenever
their aggregate render mass fits the budget, so an $r$-local preservation statement follows for the
max-hop analytical order. That statement and its proof are in Supplementary Material, Section~S2.
The evaluated implementation instead prioritises edges by minimum endpoint distance. Since the
proofs of Theorem~\ref{thm:preservation} and Proposition~\ref{prop:noharm} use only that every
seed-local unit precedes every non-seed-local unit, a property both edge-priority variants share
because the seed descriptions and seed-incident edges occupy the first two tiers under either
rule, the seed-local guarantees hold for the implemented min-hop renderer as well. The present
paper therefore makes no $r$-local preservation guarantee beyond seed-local evidence, and results
for containment-chain questions are empirical. Theorem~\ref{thm:preservation} and
Proposition~\ref{prop:noharm} are the formal guarantees for the implemented renderer;
Theorem~\ref{thm:preservation} is the sharper result for the seed-local tier-0/tier-1 subset, since
it requires only $D(S)\le B$ rather than the whole 1-hop mass.

\begin{theorem}[Naive rendering can fail; separation]
\label{thm:separation}
There exist $G$ and $q$ with seed-local $A(q)$ and $D(S)\le B$ such that
$A(q)\not\subseteq C(\pi_{\mathrm{nai}},B)$.
\end{theorem}
\begin{proof}
Proof in Supplementary Material, Section~S2.
\end{proof}

The separation is not merely existential; the loss is quantitative in a hub's
distinct-neighbour count.

\begin{proposition}[Quantitative hub flooding]
\label{prop:quant}
Under $\pi_{\mathrm{nai}}$, let the entity descriptions of $U_k(S)$ have lengths summing to
$\Delta_{\mathrm{desc}}$, rendered strictly before any edge. Let $N(h)$ be the set of
distinct non-self neighbours of a hub $h$, and suppose
$d(h,S)\le k-1$ and $m=|N(h)|\ge1$. Then every $v\in N(h)$ lies within hop $k$, and
\[
\Delta_{\mathrm{desc}}\ge m\,\ell,\qquad
\ell=\min_{v\in N(h)}\mathrm{len}(\mathrm{desc}(v))>0.
\]
Whenever $B<\Delta_{\mathrm{desc}}$ (in particular, whenever $B<m\,\ell$),
$C(\pi_{\mathrm{nai}},B)$ contains no edge rendering, so every edge-rendering unit in
$A(q)$ is absent. If
\[
L=\sum_{x\in U_k(S)}\mathrm{len}(x),\qquad
T=\sum_{x\in C(\pi_{\mathrm{nai}},B)}\mathrm{len}(x),
\]
then $T/L\le B/L$, and all admitted mass is description mass. If, in addition, $A(q)$ is
seed-local and $D(S)\le B$, then
$A(q)\subseteq C(\pi_{\mathrm{anc}},B)$ by the answer-preservation theorem. Consequently,
for any such query whose $A(q)$ contains at least one edge-rendering unit, naive rendering
is not retrieval-satisfying while anchored rendering is retrieval-satisfying.
\end{proposition}
\begin{proof}
Proof in Supplementary Material, Section~S2.
\end{proof}

On SmallGrid ($m{=}381$, $\ell$ tens of characters, $B{=}8{,}000$) the condition $B<m\,\ell$
holds with room to spare, matching the near-total collapse of naive rendering to S2 accuracy
$0.12$ (and to $0.00$ on the node-breaker SmallGrid-3.0, where the flooding hub sits on the
derived support graph; Section~\ref{sec:results-main}). In power-system terms, a high-degree hub, a large substation, a shared voltage level, or a
busbar, floods the budget with neighbour descriptions and truncates the incident connectivity edges,
hiding the actual topology around the seed component. This flooding is the known
supernode-truncation and lost-in-the-middle failure of long
contexts~\citep{gu2026_graphsteal_supernode, liu2023_lost_in_the_middle}; the proposition does not
discover it but ties it quantitatively to the naive descriptions-first order.

\begin{proposition}[No-harm dominance]
\label{prop:noharm}
In the seed-local, $D(S)\le B$ regime: every query retrieval-satisfied by
$\pi_{\mathrm{nai}}$ is retrieval-satisfied by $\pi_{\mathrm{anc}}$; and when the whole
universe fits ($\sum_{x\in U_k(S)}\mathrm{len}(x)\le B$) the two contexts contain the same
unit set.
\end{proposition}
\begin{proof}
If the universe fits, both orders render all of $U_k(S)$: identical sets. Otherwise
Theorem~\ref{thm:preservation} already gives $A(q)\subseteq C(\pi_{\mathrm{anc}},B)$
independently of $\pi_{\mathrm{nai}}$'s outcome; containment of satisfied-query sets
follows; the inclusion can be strict, as witnessed by the instances of Theorem~\ref{thm:separation}.
\end{proof}

\begin{corollary}
Restricted to seed-local answers satisfying the precondition $D(S)\le B$, seed-anchored
retrieval is never worse than naive rendering in answer \emph{presence} in the rendered
context, so a presence regression of the kind we observed for a threshold heuristic
(Supplementary Material, Section~S1) is impossible by construction. This says nothing about
non-seed-local queries whose answers lie outside the $k$-local neighbourhood (multi-hop
answers, and the S3 to S5 aggregation and superlative strata), and it does not guarantee that
the reader extracts an answer that is present.
\end{corollary}

\textbf{Scope.} The guarantees are retrieval-level: they bound answer \emph{presence}, a
necessary condition, not whether the reader extracts it (empirical, Section~\ref{sec:results}).
The formal guarantee covers S1 and S2 (Theorem~\ref{thm:preservation}). Results for S3
containment chains are empirical, because the implemented min-hop edge priority does not admit the
radius argument. S4 aggregation answers are not single units and S5
superlatives name no seed, so neither is seed-local, and we report both as the
frontier of the guarantee.

\section{Method and experimental design}
\label{sec:method}

Our study fixes every stage of a budget-bounded retrieval pipeline and varies only the
order in which render units are placed into the context window. The methodological advance is the
render-boundary guarantee: a deterministic ordering that introduces no additional tuned or learned parameters and that, under the checkable budget
precondition, provably preserves the seed-local answer-bearing render units whose render mass fits
the budget, so that a high-degree hub cannot evict them; when that precondition fails the ordering
changes which material survives truncation rather than preventing overflow. This section states the
seed-anchored retrieval procedure, briefly notes the abandoned predecessor that led us to it,
and specifies the networks, arms, questions, scoring, and statistics used to evaluate it.

\paragraph{What the benchmark measures}
The primary objective is to characterize budget-bounded retrieval and ordering effects, not
general-purpose question answering. The questions are template-generated deterministically from the
same CGMES export that supplies the retrieval corpus, so 83\% of the gold answers appear verbatim in
the corpus. The rate is measured by one rule: the exact \texttt{gold\_answer} string occurs in that
network's \texttt{corpus.txt}, counted over the 300 S1 and S2 items. We therefore read the
benchmark as a retrieval-under-budget probe, gauging whether an
answer-bearing render unit survives truncation, not whether a reader can paraphrase or reason over
free text. Because the benchmark is self-referential, we defer general-QA claims and treat two external checks as
the load-bearing robustness evidence: a non-verbatim derived-answer subset whose gold is absent from
the corpus, and a 100-question human-phrased stress test, both examined as threats in the
Limitations.

\subsection{Seed-anchored retrieval}
Algorithm~1 (Figure~\ref{alg:anchored}) realizes the anchored ordering $\pi_{\mathrm{anc}}$. A
breadth-first pass labels every node with its hop distance from the query seeds, and render
units are then emitted in the priority order established in Section~\ref{sec:theory} while a
greedy fill consumes the budget. The procedure is deterministic and introduces no additional method-specific tuned or learned parameters beyond the shared hop bound and context budget (no
learned parameters, no degree threshold, no LLM tokens), and it changes only the rendering order of
the baseline retriever. The graph, the seeding rule, the hop bound, the budget, the reader, and the
scorer are held identical to the naive arm, so any difference in outcome is attributable to order
alone. The shared name-match seeding assumes clean identifiers;
on noisy real-world inputs it can miss aliased or misspelled operator names, a robustness gap we
do not address here.

\begin{figure}[t]
\begin{center}
\fbox{\begin{minipage}{.92\linewidth}
\textbf{Algorithm 1: seed-anchored budget-bounded retrieval.}\\[2pt]
\emph{Input:} graph $(V,E)$, query $q$, budget $B$, hop bound $k$.
\begin{enumerate}\itemsep2pt
\item $S :=$ name-match seeds of $q$ (shared rule, all arms).
\item Compute hop distances $d(\cdot,S)$ by breadth-first search up to $k$.
\item Form the unit stream: seed descriptions; seed-incident edges; 1-hop descriptions;
1-hop edges; 2-hop descriptions; 2-hop edges (ties by identifier).
\item Emit units in order; stop before the first unit that would exceed $B$.
\end{enumerate}
\end{minipage}}
\end{center}
\vspace{-4pt}
\caption{Seed-anchored budget-bounded retrieval (Algorithm 1).}
\label{alg:anchored}
\end{figure}

The budget is a render-length budget $B$, measured throughout the experiments in
characters and set to $B = 8{,}000$. We use characters rather than model tokens because the
renderer is deterministic and tokenizer-independent, so the same 8{,}000-character cap
applies uniformly across arms and readers and does not import the vocabulary of any single
model into the analysis.

\subsection{A refuted predecessor}
Anchoring was not our first attempt: a fixed degree threshold that refused to expand through
high-degree nodes lifted the collapsed stratum on the development network but then fell below the
naive baseline on a previously unseen network, since no single degree cutoff separated helpful hubs
from legitimate nodes that must not be blocked. We therefore adopted the ordering of Algorithm~1,
which introduces no additional method-specific tuned or learned parameters beyond the shared hop bound and context budget and whose no-harm property makes regressions of this kind
impossible within the guaranteed regime; the full episode is reported in the Supplementary Material,
Section~S1.

\subsection{Networks, arms, and questions}
\textbf{Networks.} We evaluate on six ENTSO-E CGMES Conformity Assessment test configurations from
the powsybl-core mirror. Three are CGMES 2.4.15: MicroGrid (233 kept objects), MiniGrid (462), and
SmallGrid (1{,}398). Three are their CGMES 3.0 counterparts: MicroGrid-3.0 (317), MiniGrid-3.0 (455),
and the structurally distinct node-breaker SmallGrid-3.0 (5{,}315). The 2.4.15 exports use a
bus-branch topology whereas the 3.0 exports, most sharply SmallGrid-3.0, use the more detailed
node-breaker representation; the same physical grid is exchanged under either encoding, and, as we
show, which representation overflows the context budget flips with that choice. The 2.4.15 trio is the
development set on which the method and seeder were designed and debugged; the CGMES 3.0 trio consists
of the same underlying networks re-encoded in the newer CGMES version~\citep{entsoe_cgmes30} and is
held out from retrieval-method, seeding, hop-bound, threshold, and scorer development. Parser
compatibility fixes for nullable CGMES 3.0 fields are disclosed separately: parsing the 3.0
exports required adding \texttt{None}-guards to the CGMES parser for fields that the node-breaker
encoding may legitimately omit, and those guards change the emitted corpus text. The change is
recorded in the deviations register. It is a parsing-robustness fix rather than a method,
threshold or hop-bound decision, but the 3.0 data did prompt a code change, so the held-out claim
is scoped accordingly rather than stated absolutely.

\textbf{Arms.} The graph arms are $G_{\mathrm{STD}}$ (raw CIM graph) and $G_{\mathrm{HYB}}$ (hybrid
derived support graph), each under the naive dump-then-truncate order; seed-anchored retrieval over
$G_{\mathrm{HYB}}$, our method; and $G_{\mathrm{LLM}}$, a LightRAG~1.5.1 extraction (gleaning 2,
512-token chunks, local MiniLM embeddings~\citep{wang2020_minilm}) read under the naive order. To de-confound graph source
from rendering order we add $G_{\mathrm{LLM}+\mathrm{ANCHORED}}$, the same extracted graph under the
seed-anchored order, on the three networks with a $G_{\mathrm{LLM}}$ graph; and to test a second
framework, Microsoft GraphRAG (graphrag~3.1.0) built with the identical open-weight extractor
(gpt-oss-120b~\citep{openai2025_gptoss}), chunking (512/64), and gleaning (2), read under both orders on the same three
networks; and, as a third extraction paradigm, HippoRAG 2 (2.0.0a4)~\citep{gutirrez2025_hipporag2} built via its OpenIE module
under the same open-weight extractor and read under both orders on the same three networks. Three corpus arms operate over the deterministic text rendering: flat lexical RAG (tf-idf~\citep{salton1988}),
vector RAG (MiniLM cosine), and a long-context arm filling the budget with the corpus prefix. Every
arm shares the same seeding rule, 8{,}000-character cap, frozen reader, and deterministic scorer.

\textbf{Questions.} Each network contributes 50 pre-registered items in strata S1 and S2
(25 single-hop attribute lookups and 25 seed-local line-to-node topology items), plus 17 to
24 extended items spanning S3 (containment chains), S4 (aggregation counts), and S5 (global
superlatives). Every item is generated deterministically from the CGMES export with a
machine-checkable gold specification and unique-name filtering, yielding 129 extended items
in total.

\subsection{Scoring and statistics}
The primary scorer is deterministic, driven by a typed gold specification per item: a numeric answer
matches within a relative tolerance of 1\% of the gold value; an entity answer requires the gold
identifier to be \emph{asserted} as the answer, not merely to appear somewhere in the response, so
a reply that restates the options and then asserts a different entity scores incorrect, while a
reply that asserts the gold after a copula scores correct, and matching is insensitive to case,
separators and the class prefix the item banks store; an entity-set answer requires every gold entity; and a
count answer the exact integer. Abstentions are scored incorrect for every arm, and the scorer is
released with the artifact. Two LLM judges corroborate it. The first is gpt-oss-120b, which is the
same model used as the open-weight extractor in some evaluated arms and therefore provides only
partial independence; it scores all rows of the decisive comparisons. The second is Claude Haiku
4.5~\citep{anthropic2025_haiku45},
a genuinely cross-family judge, which gives a secondary opinion. Agreement is reported per judge
rather than pooled, and at both the per-arm and the per-stratum level, because the two levels give
different ranges and a single range would obscure which is quoted. Measured against the frozen
deterministic scorer, gpt-oss-120b agrees at 0.97 to 1.00 per arm and 0.94 to 1.00 per stratum,
pooling to 0.987 over 2,050 judged rows. Claude Haiku agrees at 0.98 to 1.00 per arm on SmallGrid
and 0.96 to 0.98 on SmallGrid-3.0, spanning 0.92 to 1.00 per stratum, pooling to 0.983 over 650
judged rows. The weakest single stratum is 0.92, on the SmallGrid-3.0 multi-hop items, and is
reported rather than absorbed into a pooled range (Supplementary
Material, Section~S7).

The statistical plan follows the pre-registration, and we are explicit about which analyses it
covers because the distinction matters for how the evidence should be read.

\emph{Primary, confirmatory.} The pre-registration names a single primary test: a pooled McNemar
exact test~\citep{mcnemar1947} of naive against seed-anchored on a freshly generated item bank spanning the two
budget-binding encodings, with no multiplicity correction because there is one test. The fresh
bank was not used during development and was generated after the pre-registration was frozen; it
uses the already-known two budget-binding CGMES encodings, so it controls item-level exposure
rather than network-level exposure. A pre-fixed sensitivity condition required
the direction to hold in both encodings and both strata.

\emph{Secondary and exploratory.} The six per-network McNemar tests on the original bank are
\emph{exploratory}. They were computed after the original bank had been used throughout method
development, so they are post hoc with respect to that bank, and we report them Holm-corrected~\citep{holm1979} and
labelled as such rather than as a confirmatory family. Because the six networks are three base
topologies re-encoded in two CGMES versions, treating them as six independent tests also
understates dependence; we therefore additionally report a clustered analysis at the level of the
three independent base topologies.

We report Wald 95\% intervals on the exploratory and RealGrid contrasts, where an interval carries
interpretive weight. The primary endpoint's interval is the McNemar-Wald interval on the paired
difference; the secondary paired contrasts against the reimplemented PPR and $G_{\mathrm{LLM}}$
arms use the Newcombe method-10 paired interval~\citep{newcombe1998}, whose endpoints differ
slightly from the score-based variants of the same family. As an a priori sensitivity anchor we also report a minimum detectable effect.
For a paired McNemar test the MDE is $(z_{\alpha/2}+z_{1-\beta})\sqrt{d/n}$, where $d$ is the
discordant-pair proportion, so the value depends on an assumption about $d$ that must be stated.
Under the conservative half-discordant model $d{=}0.5$ the MDE is $0.280$ at $n{=}50$,
$\alpha{=}0.05$, and power 0.80. The pre-registration recorded $0.198$, which corresponds to
$d{=}0.25$; the discordance actually observed across the six networks has mean $0.273$, giving
$0.207$. We quote the conservative $0.280$ as a design-sensitivity benchmark, not as a second
significance threshold: statistical significance is decided by the test, and an effect can be
significant while falling below the conservative benchmark.

The reader is Qwen2.5-7B-Instruct in bfloat16 with greedy decoding under the 8{,}000-character cap,
run on a single NVIDIA GB10 Grace Blackwell superchip (DGX Spark, 128 GB unified memory, CUDA 13.0,
PyTorch 2.12). Because decoding is greedy the runs are reproducible, and the arm tables reproduced
identically on re-execution.

\section{Results on ENTSO-E CGMES test configurations}
\label{sec:results}

We present the evidence from the retrieval-level mechanism outward: the pre-registered accuracy
contrast against naive rendering, the bounded comparison against LLM-extracted GraphRAG (de-confounding
graph source from render order) and its extractor sensitivity, the overflow flip, a corpus
consistency check, and a model-derived human-phrased stress test, with further robustness and
scope studies in the Supplementary Material. Numbers are from the deterministic scorer unless noted;
the two LLM judges (gpt-oss-120b and Claude Haiku) are zero-shot with no exposure to these networks
and only corroborate it. Because the questions are template-generated from the same corpus
(Section~\ref{sec:method}), single-hop and topology accuracies are budget-bounded retrieval upper
bounds, and the 100-question human-phrased stress test (Section~\ref{sec:results-human}, method
$0.83$) is the primary non-self-referential evidence, to be read before the template-generated
accuracies below.

\subsection{Mechanism: answer presence under the cap}
\label{sec:results-mech}
Before any accuracy is measured, a reader-free audit on the development SmallGrid multi-hop stratum
exposes the mechanism. Naive rendering truncated at the budget leaves the answer-bearing units in
context on 3 of 25 items, all 25 contexts at the cap; seed-anchored rendering keeps the answer
present on all 25 with a single context at the cap. The pattern holds across the run set
(Fig.~\ref{fig:caphits}): naive contexts hit the cap on 32 to 50 of 50 items per network (47 to 50 on
five of six), anchored on at most 2 of 50.

Three reader-free audits support the mechanism claim and are referred to throughout by short
labels: E1 is a deterministic structured-query baseline over the parsed model, E2 is the
evidence-scoped answer-presence audit, and E3 is the precondition audit that checks
$D(S)\le B$ per question. E2 is reported here; E1 and E3 are in Supplementary Material,
Section~S3.

The E2 audit tests the theorem's empirical prediction before reader
accuracy is measured. On the two budget-binding networks, naive rendering places the
evidence-linked gold content in context for $0.56$ of questions on SmallGrid and $0.50$ on
SmallGrid-3.0, and seed-anchored rendering for $1.00$ on both; where the budget does not bind,
both reach $1.00$. These rates are pooled over the $n{=}50$ pre-registered S1+S2 items per
network and are not multi-hop-only figures; the per-stratum split is given below and is what the
mechanism argument rests on. Under the weaker substring test, which counts an equal token anywhere in the
context, naive reaches $0.70$ on both networks. E2 is measured against each item's
evidence path: numeric attributes are matched only in the evidence-resolved field of the emitted
seed-description unit (relative tolerance $10^{-4}$), entity answers only in evidence-resolved
entity or edge units, and counts from the complete evidence-resolved member set; an equal token
elsewhere in the context does not count. A precondition audit confirms that the seed-local render
mass $D(S)$ (median 592 to 1{,}931 characters, maximum 3{,}063) satisfies $D(S)\le B$ at
$B=8{,}000$ for every question on every network, so the observed anchored value of $1.00$ is
consistent with Theorem~\ref{thm:preservation}. Hop-ordered ego expansion, degree-aware hub-pruned
BFS, and personalized PageRank have been rerun under this same evidence-scoped test, and we report
them here. On SmallGrid and SmallGrid-3.0 respectively, over the pooled $n{=}50$ items, ego
expansion reaches $1.00$ and $1.00$, hub-pruned BFS $1.00$ and $1.00$, and personalized
PageRank~\citep{page1999pagerank} $0.56$ and $0.50$, against $1.00$ and
$1.00$ for seed-anchored. Split by stratum, ego expansion, hub-pruned BFS and seed-anchored all
reach $1.00$ on S1 and on S2, whereas personalized PageRank matches naive exactly, $1.00$ on S1
and $0.12$ and $0.00$ on S2, so its pooled parity with naive is a multi-hop failure and not an
even shortfall across strata. Seed-anchored rendering is therefore not uniquely better than simple
locality priors at getting the answer into the window: ego expansion and hub-pruned BFS match it
on both budget-binding networks. What distinguishes seed-anchored rendering is not presence
superiority but the checkable a priori guarantee and the fact that the method introduces no additional method-specific tuned or learned parameters beyond the shared hop bound and context budget: ego expansion and
hub-pruned BFS reach the same presence empirically, with a tuned radius and a tuned degree
threshold respectively and with no precondition that can be checked before the render is
performed.

Presence is necessary, not sufficient, and separating the two bounds the claim. The E2 rates quoted
above are pooled over the $n{=}50$ S1+S2 items per network; the multi-hop stratum must be read
separately. Split by stratum, naive rendering places the evidence-linked gold content in context for
every S1 item on both budget-binding networks ($25/25$), and for $0.12$ of S2 items on SmallGrid
($3/25$) and $0.00$ on SmallGrid-3.0 ($0/25$); seed-anchored reaches $1.00$ on both strata and both
networks. An item-level join of E2 presence against frozen-scorer reader correctness shows the two
sets coincide exactly on S2: on SmallGrid the reader is correct on all three items where the
evidence survived and on none of the twenty-two where it did not, and on SmallGrid-3.0 no S2 item
retains its evidence and none is answered correctly. On this stratum the naive failure is therefore
a truncation failure and not a reader failure: where the evidence was in the window the reader used
it every time. The reader-side residual sits on S1 instead, where presence is $1.00$ but accuracy is
$0.92$ on SmallGrid and $0.84$ on SmallGrid-3.0, so two and four items respectively are answered
incorrectly from evidence that was present. We therefore report the multi-hop failure as
retrieval-only, and the single-hop residual as reader-side and outside what the theorem speaks to.
The E1 structured-query analysis and E3 precondition audit are in Supplementary Material,
Section~S3.

\begin{figure}[t]
\centering
\includegraphics[width=.90\linewidth]{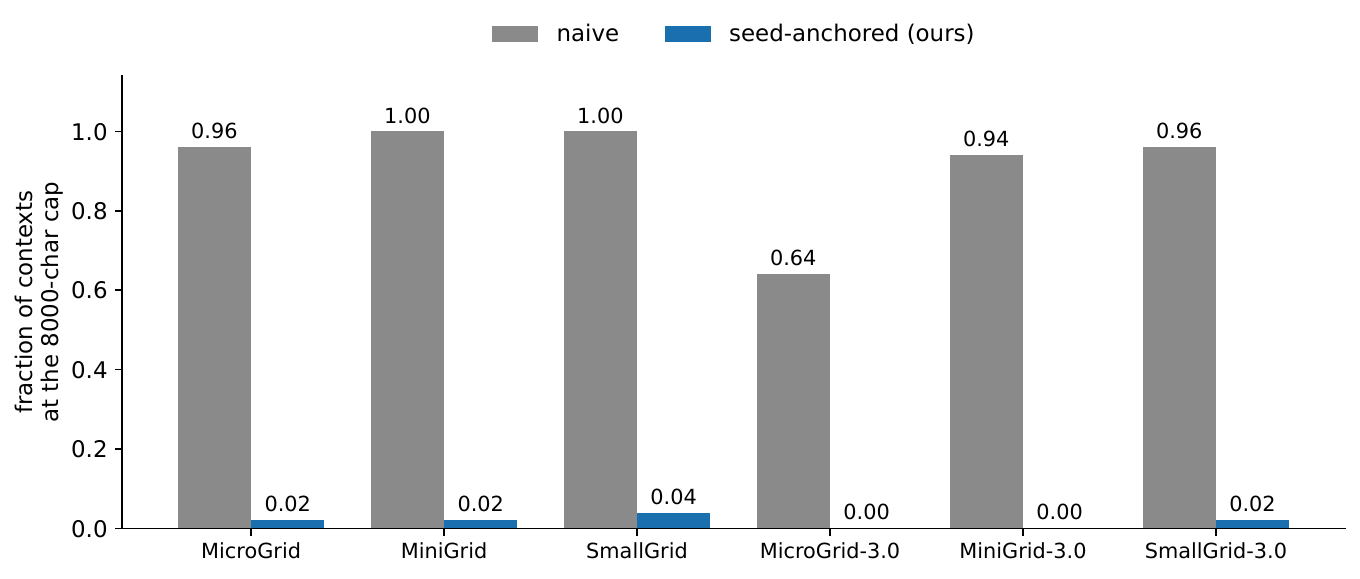}
\caption{Fraction of contexts truncated at the 8{,}000-character budget for the ENTSO-E CGMES test
networks under naive vs seed-anchored retrieval, per network (S1+S2). Bar values are the cap-hit
counts of Table~\ref{tab:main} expressed as fractions of $n{=}50$.}
\label{fig:caphits}
\end{figure}

\subsection{Seed-anchored retrieval improves CGMES power-system QA under context budgets}
\label{sec:results-main}
\paragraph{The pre-registered primary endpoint}
The confirmatory result is the pooled test on the fresh item bank, generated after the
pre-registration was frozen and with the original bank excluded. Across the two budget-binding
encodings ($n{=}100$), naive rendering scores $0.450$ and seed-anchored $0.970$, a paired
difference of $+0.520$ (95\% CI $[+0.415,+0.625]$; discordant pairs $b{=}2$, $c{=}54$; one-sided
McNemar exact $p{=}2.2\times10^{-14}$). Its scope should be read precisely: the fresh bank is fresh
at the item level and is drawn on the two SmallGrid encodings, which are the development topology
family, so it controls template and item exposure but not network-level exposure. Evidence on a
topology never seen during development is the RealGrid replication below, which is
deterministic-arms-only. A pre-fixed comparability check found the worst
fresh-versus-original accuracy gap to be 8.0 percentage points, inside the 10-point threshold
declared in advance, so the confirmatory claim carries no comparability caveat.
Table~\ref{tab:freshbank} gives the endpoint by encoding and stratum. Anchored exceeded naive in
each of the four encoding-by-stratum cells, which is the pre-fixed sensitivity condition. The
cells demonstrate direction and robustness; the primary test is the pooled one, and no cell-level
$p$-value is reported because none was pre-specified.

\begin{table}[t]\centering\small
\caption{The pre-registered confirmatory endpoint, by encoding and stratum. The pooled row is the
primary test. Two readings should be kept in view: the endpoint is specific to one 7B reader at
the 8{,}000-character budget, and gold strings occur verbatim in the corpus for $0.50$ of
SmallGrid items and $1.00$ of SmallGrid-3.0 items, so on the latter the measurement is closer to
in-context retrievability than to question answering.}
\label{tab:freshbank}
\begin{tabular}{llrccc}\hline
Encoding & Stratum & $n$ & naive & seed-anchored & paired difference \\ \hline
SmallGrid 2.4.15 & S1 & 25 & 0.80 & 0.96 & $+0.16$ \\
SmallGrid 2.4.15 & S2 & 25 & 0.08 & 1.00 & $+0.92$ \\
SmallGrid-3.0 & S1 & 25 & 0.92 & 0.96 & $+0.04$ \\
SmallGrid-3.0 & S2 & 25 & 0.00 & 0.96 & $+0.96$ \\ \hline
Pooled & both & 100 & 0.450 & 0.970 & $+0.520$ \\
\hline\end{tabular}\end{table}

\paragraph{Exploratory per-network contrasts on the original bank}
Table~\ref{tab:main} and Fig.~\ref{fig:main} report the six per-network contrasts. These are
exploratory: the original bank was used throughout method development, so they are post hoc with
respect to it. The point estimate
favours seed-anchored rendering on every one of the six networks. After Holm correction, the gain
is significant on three networks: MicroGrid rises from $0.68$ to $0.92$
($p{=}7.3\times10^{-3}$), SmallGrid from $0.52$ to $0.98$
($p{=}1.2\times10^{-6}$), and held-out SmallGrid-3.0 from $0.42$ to $0.96$
($p{=}8.9\times10^{-8}$), with respective gains of $+0.24$, $+0.46$ and $+0.54$. The conservative
design-sensitivity benchmark is $0.280$; SmallGrid and SmallGrid-3.0 exceed it and MicroGrid does
not, which bears on how confidently the MicroGrid effect size should be read but does not change
its significance. Only SmallGrid and SmallGrid-3.0 are
budget-binding; MicroGrid is significant but not budget-binding.

Because the six networks are three base topologies re-encoded in two CGMES versions, we also
report the analysis clustered at the level of the three independent base topologies, which
removes the pseudo-replication the six-test framing invites. Two of three are significant after
Holm correction over three tests: MicroGrid $+0.17$ ($p{=}4.4\times10^{-4}$) and SmallGrid
$+0.50$ ($p{=}5.3\times10^{-15}$); MiniGrid $+0.07$ is not significant ($p{=}0.065$). The
defensible summary of the exploratory analysis is therefore two of three independent base
topologies, not three of six networks. On the two budget-binding
networks, SmallGrid multi-hop S2 accuracy rises from $0.12$ to $0.96$, while on
SmallGrid-3.0 naive rendering answers no multi-hop item (S2 accuracy $0.00$) and anchored restores
S2 to $0.92$. Three comparisons remain: MiniGrid ($+0.04$, $p{=}6.3\times10^{-1}$),
MicroGrid-3.0 ($+0.10$, $p{=}3.8\times10^{-1}$), and MiniGrid-3.0
($+0.10$, $p{=}3.8\times10^{-1}$). All three fall below the MDE and are underpowered, so we treat
them as inconclusive rather than as wins. No network shows a significant negative difference. For the
three remaining networks, end-to-end differences are inconclusive; the theorem establishes only
retrieval-level answer presence under its preconditions and does not establish accuracy
non-inferiority. These six per-network McNemar tests are exploratory, as set out in
Section~\ref{sec:method}; the confirmatory result is the fresh-bank endpoint reported above.

\begin{table}[t]\centering\small
\caption{Seed-anchored vs naive retrieval on six ENTSO-E CGMES test networks (S1+S2, $n{=}50$ per
network, corrected unit-aware deterministic scoring). Items are template-generated CGMES-based QA.
The naive arm here is $G_{\mathrm{HYB}}$ under dump-then-truncate rendering, not
$G_{\mathrm{STD}}$+naive, which is reported separately in Table~\ref{tab:arms}.
Cap hits: contexts reaching the 8000-character budget. $p$: McNemar exact, Holm-corrected across
the six networks. The CGMES 3.0 networks are held out from retrieval-method, seeding, hop-bound,
threshold and scorer development; parser compatibility fixes for nullable 3.0 fields are disclosed
in Section~\ref{sec:method}. These six contrasts are exploratory, not the confirmatory endpoint.}
\label{tab:main}
\resizebox{\linewidth}{!}{\begin{tabular}{lccccc}\hline
Network & naive & anchored & cap hits (naive) & cap hits (anch.) & Holm $p$ \\ \hline
MicroGrid & 0.68 & 0.92 & 48/50 & 1/50 & 7.3e-03 \\
MiniGrid & 0.88 & 0.92 & 50/50 & 1/50 & 6.3e-01 \\
SmallGrid & 0.52 & 0.98 & 50/50 & 2/50 & 1.2e-06 \\
MicroGrid-3.0 & 0.88 & 0.98 & 32/50 & 0/50 & 3.8e-01 \\
MiniGrid-3.0 & 0.86 & 0.96 & 47/50 & 0/50 & 3.8e-01 \\
SmallGrid-3.0 & 0.42 & 0.96 & 48/50 & 1/50 & 8.9e-08 \\
\hline\end{tabular}}\end{table}

\paragraph{An independent topology, scoped to attribute lookup}
RealGrid is an ENTSO-E CGMES 2.4.15 conformity configuration of $72{,}418$ objects on which naive
rendering saturates the $8{,}000$-character budget for every question. Its multi-hop stratum has
been \emph{retired} and is not reported. A gold-uniqueness audit found that none of its 25 S2 golds
identifies a unique referent: the topological-node names used as gold labels are shared by up to
$5{,}488$ nodes, and the renderer emits names rather than machine identifiers, so no reader can
assert a unique answer. The stratum does not measure what it claims to measure, so it is removed
rather than narrowed, together with two subgraph-sample comparisons that rested on it.

What survives is the attribute-lookup stratum, re-derived under the frozen scorer. On S1
($n{=}25$) all three deterministic arms reach $1.000$: $G_{\mathrm{STD}}$+naive,
$G_{\mathrm{HYB}}$+naive, and seed-anchored alike. RealGrid therefore supports a \emph{no-harm}
result on an independent topology, consistent with
Proposition~\ref{prop:noharm}, and not a dominance result. LLM-extracted graphs were not rebuilt at
this scale, estimated at ${\sim}34$M tokens for LightRAG alone.

All standards-native arms use the deterministic CGMES parser and
spend zero LLM tokens on graph construction, whereas every arm built by LightRAG, Microsoft
GraphRAG, or HippoRAG pays an LLM extraction cost, so \texttt{G\_HYB+ANCHORED} is a deterministic
retriever rather than an LLM-enhanced one.

\begin{figure}[t]
\centering
\includegraphics[width=.90\linewidth]{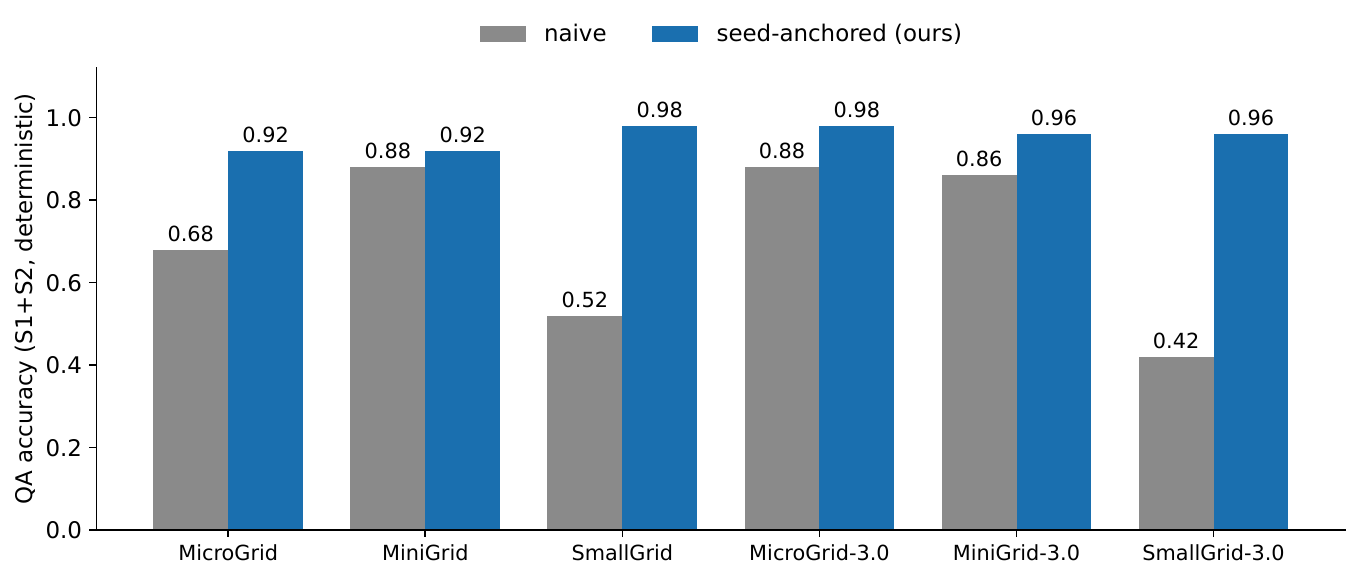}
\caption{Seed-anchored vs naive retrieval, overall S1+S2 accuracy per network
(corrected unit-aware deterministic scoring, $n{=}50$ each). Bar values are those of
Table~\ref{tab:main}. CGMES 3.0 networks are held out.}
\label{fig:main}
\end{figure}

\subsection{Against LLM-extracted GraphRAG}
\label{sec:results-gllm}
Table~\ref{tab:arms} and Fig.~\ref{fig:tokens} set seed-anchored retrieval against LLM-extracted
GraphRAG from three frameworks spanning three extraction paradigms (LightRAG keyword-relations,
Microsoft GraphRAG entity-descriptions, and HippoRAG OpenIE triples), taking LightRAG first. On the
development SmallGrid the two arms land in exact statistical parity ($0.98$ vs $0.98$; paired
difference $0.000$, 95\% CI $\pm 0.055$), yet the extracted graph cost 594{,}617 tokens to build and
the anchored arm none. On the held-out SmallGrid-3.0 the open-weight gpt-oss-120b extraction
(754{,}781 tokens) reaches $0.74$ against $0.96$ for anchored ($+0.22$, 95\% CI $[+0.08,+0.36]$,
nominal $p{=}7.4\times10^{-3}$, secondary exploratory); on MicroGrid anchored again leads ($0.92$ vs
$0.78$).

Reading the LLM-extracted graph only under the naive order confounds graph source with render order,
so we added \texttt{G\_LLM+ANCHORED}: the seed-anchored order on the LLM-extracted graph, identical
reader and scorer, on the three networks with a $G_{\mathrm{LLM}}$ graph (Table~\ref{tab:arms}). Two
findings follow. Anchored ordering is a general lever, improving the extracted graph too by $+0.06$ on
held-out SmallGrid-3.0 ($0.74$ to $0.80$) and $+0.04$ on MicroGrid ($0.78$ to $0.82$). The lever is
not uniform across frameworks: on MicroGrid, which does not bind the budget, Microsoft GraphRAG
moves from $0.72$ to $0.70$ and HippoRAG from $0.80$ to $0.78$ under the same order
(Table~\ref{tab:arms}), so the ordering gain is non-negative for all three extraction frameworks
only on the two budget-binding networks. Yet at fixed
anchored order the standards-native graph still leads ($0.96$ vs $0.80$ held-out; $0.92$ vs $0.82$
MicroGrid), a graph-source and fidelity effect de-confounded from render order. On development
SmallGrid the seed-anchored standards-native arm and both LightRAG arms, naive and anchored, all
reach $0.98$, so both levers wash out there; Microsoft GraphRAG ($0.40$/$0.46$) and HippoRAG
($0.80$/$0.86$) do not, and the source effect appears only once the network binds the budget.

\textbf{A second extraction framework.} To test whether the against-\linebreak
extraction result is specific to
LightRAG, we ran Microsoft GraphRAG (graphrag 3.1.0) under the identical pipeline (same open-weight
extractor gpt-oss-120b, matched chunking and gleaning, CGMES-appropriate entity types, same reader,
seeder, and scorer), adding its two arms to Table~\ref{tab:arms}. On both budget-binding networks it
scored below LightRAG and far below the zero-token anchored method (SmallGrid $0.40$ naive, $0.46$
anchored; SmallGrid-3.0 $0.60$, $0.70$; against $0.98$ and $0.96$ for anchored), at $2.3\times$ to
$2.7\times$ LightRAG's build cost (235{,}312 to 1{,}703{,}707 tokens). The mechanism is an
inconsistent, partly-merged entity set in which a queried component folded into a generic node has no
discrete entity to anchor: S2 connectivity survives on the node-breaker 3.0 network ($0.96$) while S1
per-component attributes do not ($0.24$), the fidelity loss extraction incurs elsewhere; LightRAG
under the identical seeder scored $0.98$ on SmallGrid, so this is a property of the extraction, not
the seeder. The full entity-merge walkthrough is in Supplementary Material, Section~S7.

\textbf{A third extraction paradigm.} HippoRAG (OpenIE triples, the cheapest of the three at
31{,}890 to 206{,}254 tokens) confirms the render-order mechanism on a third, independently extracted
graph: on held-out SmallGrid-3.0 its OpenIE graph carries a flooding hub under the naive order
($0.28$) that seed-anchored rescues to $0.78$ ($+0.50$), the no-harm ordering effect of
Theorem~\ref{thm:preservation} and Proposition~\ref{prop:noharm}, while the zero-token standards-native
arm still leads ($0.96$ vs $0.78$).

\paragraph{A framework's own retrieval scheme, reimplemented}
Reading an extracted graph under our shared renderer could understate a framework that ships its
own retriever, so we reimplemented HippoRAG's personalized-PageRank retrieval over its own
extracted graph. The packaged implementation could not be run on the arm64 evaluation host, so
this is a reimplementation and we label it as such. On SmallGrid it reaches $0.86$ against $0.98$
for seed-anchored (paired difference $+0.12$, 95\% CI $[+0.031,+0.237]$, McNemar exact
$p{=}0.031$); on held-out SmallGrid-3.0 it reaches $0.78$ against $0.96$ ($+0.18$, 95\% CI
$[+0.041,+0.319]$, $p{=}0.023$). Native PPR recovers exactly what HippoRAG's graph reaches under
our shared anchored renderer ($0.86$ and $0.78$ on the two networks), which is direct evidence
that the shared-renderer comparison is not disadvantaging the framework: the gap is in the graph,
not in our rendering of it. The two remaining frameworks' packaged retrievers were not run.

The reading is deliberately bounded: on every network, at zero LLM graph-construction tokens the
standards-native \emph{seed-anchored} arm is never behind any extracted graph under either render
order. The claim is specific to that arm. The standards-native arms read under the naive order are
behind on the budget-binding networks, which is the failure this paper is about:
$G_{\mathrm{HYB}}$ under naive rendering reaches $0.52$ on SmallGrid against LightRAG's $0.98$,
and $0.42$ on SmallGrid-3.0 against $0.74$ (Table~\ref{tab:arms}). So on this
CIM/CGMES benchmark, conditional on the three tested frameworks, extractor, reader,
and budget, LLM extraction buys no retrieval advantage. The numbers are doubly conditional: the
development-network tie depended on a frontier proprietary extractor (swapping it drops
$G_{\mathrm{LLM}}$ to $0.66$, Section~\ref{sec:results-extractor}) and absolute accuracies are
reader-dependent (Section~S4); the robust claim is the ordering under fixed cost, not the level.

\begin{table}[t]\centering\small
\caption{Full arm comparison (overall S1+S2 accuracy, deterministic retrieval; shared Qwen2.5-7B
reader, seeder, and scorer), grouped by graph source across three extraction frameworks. Standards-native arms cost zero LLM
graph-construction tokens; the three extraction frameworks pay the per-network build cost in each panel
heading (token triples list MicroGrid / SmallGrid / SmallGrid-3.0; LightRAG and HippoRAG counts are
taken from their own run logs, the Microsoft GraphRAG counts from its indexing-engine log, which is
not part of the deposited artifact; Microsoft GraphRAG costs 2.3x to
2.7x LightRAG, HippoRAG is the cheapest). All three extracted graphs are read under the identical shared render pipeline, isolating
graph source under matched retrieval; the anchored rows apply the seed-anchored order to the
extracted graph, and with that order fixed the standards-native $G_{\mathrm{HYB}}$ still leads on the
budget-binding networks. SmallGrid-3.0 is held out; its extracted graphs used an open-weight
extractor (gpt-oss-120b). Cells are reported under the corrected unit-aware
deterministic scorer, the same frozen scoring rule used in Table~\ref{tab:main}, with one
exception: the MicroGrid $G_{\mathrm{STD}}$+naive and $G_{\mathrm{LLM}}$+naive entries are stored
summary values for which no row-level predictions were retained, so they could not be re-derived
under the frozen scorer and are carried at their stored-scorer values.}
\label{tab:arms}
\begin{tabular}{lccc}\hline
Arm & MicroGrid & SmallGrid & SmallGrid-3.0 (held-out) \\ \hline
\multicolumn{4}{l}{\emph{Standards-native (zero graph-construction tokens)}} \\
$G_{\mathrm{STD}}$, naive (raw CIM) & 0.84 & 0.48 & 0.94 \\
$G_{\mathrm{HYB}}$, naive (support) & 0.68 & 0.52 & 0.42 \\
$G_{\mathrm{HYB}}$, anchored (ours) & \textbf{0.92} & \textbf{0.98} & \textbf{0.96} \\
\multicolumn{4}{l}{\emph{LightRAG extraction (87{,}962 / 594{,}617 / 754{,}781 tokens)}} \\
$G_{\mathrm{LLM}}$, naive & 0.78 & 0.98 & 0.74 \\
$G_{\mathrm{LLM}}$, anchored & 0.82 & 0.98 & 0.80 \\
\multicolumn{4}{l}{\emph{Microsoft GraphRAG extraction (235{,}312 / 1{,}492{,}123 / 1{,}703{,}707 tokens)}} \\
MS GraphRAG, naive & 0.72 & 0.40 & 0.60 \\
MS GraphRAG, anchored & 0.70 & 0.46 & 0.70 \\
\multicolumn{4}{l}{\emph{HippoRAG (OpenIE) extraction (31{,}890 / 200{,}785 / 206{,}254 tokens)}} \\
HippoRAG, naive & 0.80 & 0.80 & 0.28 \\
HippoRAG, anchored & 0.78 & 0.86 & 0.78 \\
\hline\end{tabular}\end{table}

\begin{figure}[t]
\centering
\includegraphics[width=.80\linewidth]{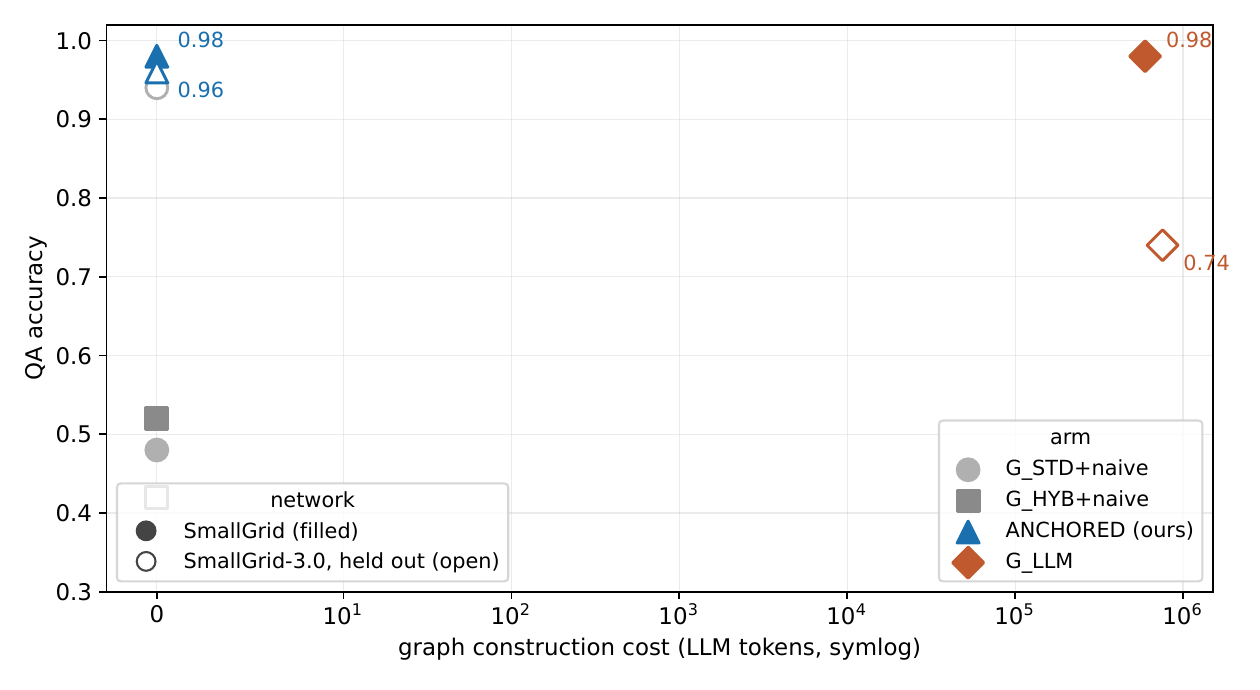}
\caption{Accuracy vs graph-construction cost on the two budget-binding networks. All
deterministic arms sit at zero tokens; LLM extraction pays $5.9\times10^5$ to
$7.5\times10^5$ tokens for parity at best.}
\label{fig:tokens}
\end{figure}

\subsection{Extractor-controlled comparison}
\label{sec:results-extractor}
The development and held-out $G_{\mathrm{LLM}}$ graphs used different extractors (proprietary on
development, open-weight gpt-oss-120b on held-out), confounding the held-out advantage with the
extractor. Rebuilding the development $G_{\mathrm{LLM}}$ on SmallGrid with the held-out extractor,
varying only that (Table~\ref{tab:extractor}, single reader, scorer, seeder), drops $G_{\mathrm{LLM}}$
from $0.98$ to $0.66$ (a 691{,}089-token extraction, near the held-out gpt-oss level of $0.74$) while
the zero-token anchored arm stays $0.98$. The held-out margin is thus extractor-conditional (the
development-network parity required a frontier proprietary extractor at 594{,}617 tokens), and the
de-confounding also puts seed-anchored $+0.32$ ahead on the development network, so that result is not
a network or version artifact.

\begin{table}[t]\centering\small
\caption{Extractor-controlled comparison on the development SmallGrid (S1+S2, $n{=}50$,
one reader/scorer/seeder). The two $G_{\mathrm{LLM}}$ rows differ only in the extractor
that built the graph. Holding the extractor at the same open-weight model used for the
held-out network drops $G_{\mathrm{LLM}}$ to the held-out level, while the zero-token
anchored arm is unchanged: the held-out margin is extractor-conditional, not a
network artifact.}
\label{tab:extractor}
\begin{tabular}{llcc}\hline
Arm & Graph extractor & build tokens & accuracy \\ \hline
anchored (ours) & none (parsed CGMES) & 0 & \textbf{0.98} \\
$G_{\mathrm{LLM}}$ & Azure gpt-5.1~\citep{openai2025_gpt51} (proprietary) & 594,617 & 0.98 \\
$G_{\mathrm{LLM}}$ & gpt-oss-120b (open-weight) & 691,089 & 0.66 \\
\hline\end{tabular}\end{table}

\subsection{Budget overflow is representation-dependent}
\label{sec:results-flip}

\begin{figure}[t]
\centering
\includegraphics[width=0.92\linewidth]{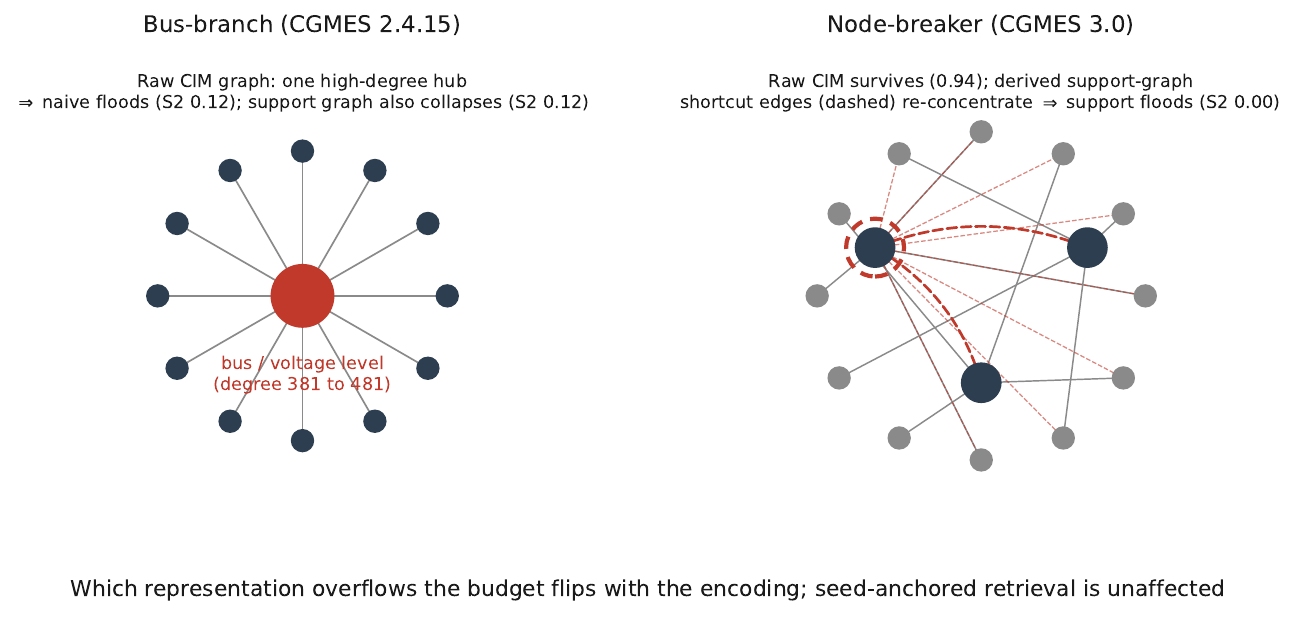}
\caption{Why budget overflow flips with the CGMES encoding. In the bus-branch encoding
(2.4.15) equipment attach directly to a high-degree bus, so the raw CIM graph carries the
flooding hub and naive rendering loses the answer edges; the derived support graph is comparatively
less affected (0.52 vs 0.48 overall) but collapses under naive rendering too (S2 0.12), so both
graphs fail here. In
the node-breaker encoding (3.0) equipment attach through breakers and connectivity nodes, which
spreads raw-graph degree and lets the raw CIM graph survive (0.94), but the support graph's derived
shortcut edges (dashed) re-concentrate degree and flood instead. Seed-anchored rendering renders
seed-incident edges first regardless of representation and is unaffected in either encoding.}
\label{fig:flip}
\end{figure}

Table~\ref{tab:arms} hides a structural surprise. On the bus-branch SmallGrid the raw CIM graph
$G_{\mathrm{STD}}$ collapses under naive rendering (0.48), while on the node-breaker SmallGrid-3.0
that same raw graph survives (0.94) and instead the derived support graph $G_{\mathrm{HYB}}$ collapses
(0.42, S2 0.00), because the node-breaker encoding spreads two-hop neighbourhoods more thinly whereas
the support graph's derived shortcut edges concentrate them (Figure~\ref{fig:flip}). Which
representation overflows the budget is thus an artifact of the encoding, and the failing arm flips
unpredictably between versions of the same network. Seed-anchored retrieval is the only arm uniformly
high across every cell, the operational content of the no-harm guarantee: a deployer need not predict
which graph-encoding combination will overflow.

\subsection{Corpus baselines: an internal consistency check}
\label{sec:results-corpus}
A third family of arms moves from graphs to the deterministic \emph{text} rendering of the
standard. Flat lexical RAG is strong on S1/S2 everywhere (0.86 to 0.98, including 0.98 on
SmallGrid); vector RAG degrades on the large networks (0.58 and 0.68, MiniLM~\citep{wang2020_minilm} retrieval failing at
scale); long-context collapses entirely there (0.00: a corpus of roughly 85{,}000 characters
truncated at 8{,}000 loses most answers, the textual analogue of Theorem~\ref{thm:separation}). Two
caveats keep the flat-RAG result honest: 249 of the 300 S1/S2 gold answers (83\%) appear verbatim,
measured as exact occurrence of the \texttt{gold\_answer} string in the network's
\texttt{corpus.txt}, so
lexical matching enjoys a surface-form advantage that will not transfer to heterogeneous field
documentation, and the corpus is itself a deterministic standards-native artifact. We therefore read
flat-RAG not as evidence of generalization but as a one-directional consistency check: under the
conditions stated above, a faithful standards-native representation with budget-aware retrieval
matches or exceeds the extracted graphs in either graph or text form (generalization beyond the
corpus is tested in Supplementary Material, Section~S4).

The per-network corpus-baseline accuracies are tabulated in Supplementary Material, Section~S3.

\subsection{A model-derived, human-phrased stress test}
\label{sec:results-human}
To probe outside the self-referential regime we use a human-authored set of 100 operator-style
questions on SmallGrid whose wording the templates never generate, spanning eight families: attribute
reads (S1), forward and reverse topology (S2), shared-node intersection (R1), unit conversion (D1),
comparison (D2), two-fact arithmetic (D3), a boolean inequality (D4), and thirteen adversarial traps
(false premises, name-node confusions, near-duplicate identifiers, distractors). Every gold answer is
resolved deterministically from the parsed model and all 100 reconcile, so no language model touches
ground truth: the set is model-derived but human-phrased, stressing robustness to non-templated
adversarial wording rather than serving as an externally annotated benchmark (field-authored gold is
future work). Table~\ref{tab:human} gives accuracy by family.

The result confirms the thesis on non-templated questions and marks its boundary honestly. The
proposed method ($G_{\mathrm{HYB}}$ with seed-anchored rendering) leads overall at $0.83$, against
$0.69$ for the LLM-extracted arm and $0.42$ for naive rendering, and the gaps fall where the theory
predicts. All probe values quoted here are frozen-scorer values and are read directly from
Table~\ref{tab:human}. Seed-anchored rendering reaches higher accuracy on topology ($0.89$ vs $0.11$ naive) and shared-node
intersection ($1.00$ vs $0.25$), where two lines' neighbourhoods share the budget and naive
rendering truncates one. The LLM-extracted arm, strong on single-attribute reads ($0.96$) and unit
conversion, loses the exact topology and values needed for intersection ($0.00$), comparison
($0.22$), and two-fact arithmetic ($0.29$), reinforcing that a reconstruction is not a substitute
for the parsed model. On the thirteen adversarial traps the method holds $0.77$ against $0.31$ for
naive, so the ordering advantage is not an artifact of clean phrasing. The probe is paired at
$n{=}100$, so we test it. Exact two-sided McNemar against the proposed method gives
$p{=}8.2\times10^{-10}$ versus $G_{\mathrm{STD}}$+naive, $p{=}3.6\times10^{-8}$ versus
$G_{\mathrm{STD}}$+anchored, and $p{=}9.4\times10^{-3}$ versus the LLM-extracted arm. Families with
$n\le4$ (R1, D1, D4) are reported descriptively as raw fractions and are not interpreted as stable
family-level estimates. The honest ceiling is
reader-side computation: the arithmetic and boolean families (D3, D4; D4 a single item) bound every
arm, because retrieving both operands is necessary but the small reader must still compute
correctly, so the residual failures sit in reader arithmetic rather than retrieval.

\begin{table}[t]\centering\small
\caption{Human-phrased robustness probe on SmallGrid (100 operator-style questions; wording written by hand and not template-generated; gold resolved deterministically from the parsed CGMES model; local reader). Per-family correctness across the decisive arms, with family sizes given in the row labels. This is a single-network robustness probe, not an independent external benchmark. All cells are reported under the corrected unit-aware, assertion-aware deterministic scorer (frozen; sha256 2d9c8cfa), the same rule used for the six-network main table, and are scored against untruncated reader responses recovered by the same-item regeneration certificate, every row of which was prefix-verified against the stored text. The $G_{\mathrm{STD}}$+ANCHORED column reports the arm the Section~6 paired tests are computed against, on the same basis as the other three columns.}
\label{tab:human}
\resizebox{\linewidth}{!}{\begin{tabular}{lcccc}\hline
Family & G\_HYB+ANCHORED & G\_STD+ANCHORED & G\_LLM+naive & G\_STD+naive \\ \hline
S1: direct attribute ($n{=}25$) & 1.00 & 0.92 & 0.96 & 0.84 \\
S2: topology ($n{=}38$) & 0.89 & 0.18 & 0.79 & 0.11 \\
R1: shared-node intersection ($n{=}4$) & 1.00 & 0.25 & 0.00 & 0.25 \\
D1: unit conversion ($n{=}3$) & 0.67 & 0.67 & 1.00 & 0.67 \\
D2: value comparison ($n{=}9$) & 0.44 & 0.78 & 0.22 & 0.67 \\
D3: two-fact arithmetic ($n{=}7$) & 0.57 & 0.29 & 0.29 & 0.43 \\
D4: boolean inequality ($n{=}1$) & 0.00 & 0.00 & 0.00 & 1.00 \\
ADV: adversarial traps ($n{=}13$) & 0.77 & 0.62 & 0.62 & 0.31 \\
\hline
overall ($n{=}100$) & 0.83 & 0.50 & 0.69 & 0.42 \\
\hline\end{tabular}}\end{table}

\paragraph{Held-out transfer on SmallGrid-3.0}
\label{sec:results-humanv3}
A second 100-question human-authored set on the held-out SmallGrid-3.0 (node-breaker, 5{,}315
objects), run at the published budget and 2-hop retrieval with the same reader and scorer and all
golds reconciled against the parsed model, shows the fidelity claim transfers to an unseen network:
the proposed method leads overall ($0.62$) and the other structural-graph arms answer the
questions ($0.53$ and $0.57$), while the LLM-descriptor arm collapses to
$0.16$, abstaining on 84 of 100
questions, its descriptor graph summarizing rather than preserving exact attributes. Testing this
probe changes how it should be read. Against the LLM-extracted arm the method is decisively ahead
(exact two-sided McNemar $p{=}2.3\times10^{-12}$), and against $G_{\mathrm{STD}}$+anchored it is
marginal ($p{=}0.049$); but against $G_{\mathrm{STD}}$+naive the difference is
\emph{not} significant ($0.62$ vs $0.57$, $b{=}12$, $c{=}7$, $p{=}0.36$). On this held-out,
mostly non-binding probe the method is therefore separable from the LLM-extracted arm but not from
the structural naive arm, which is what the absence of a binding budget predicts. We state the
consequence plainly: no held-out, non-templated comparison under a binding budget exists in this
study. The held-out probe is not budget-binding, and the budget-binding comparisons are templated,
so the two properties we would most want combined are never combined in a single evaluation. The
SmallGrid human-probe margin against the naive arm is accordingly exploratory rather than
confirmatory. This gap
is partly confounded by its different open-weight gpt-oss-120b extractor
(Section~\ref{sec:results-extractor}; Table~\ref{tab:humanv3}, per-family detail in Supplementary
Material, Section~S5). We mark the boundary plainly: here SmallGrid-3.0 is mostly non-budget-binding
(naive context averaging 2262 characters against the 8000-character budget, only 13 of 100 questions
reaching at least 98\% of it), so it does not reproduce the budget-binding dominance margin SmallGrid establishes
(Table~\ref{tab:human}) and we do not present it as a second dominance result. The two networks are
complementary: SmallGrid establishes the anchored-versus-naive dominance margin, and held-out
SmallGrid-3.0 establishes fidelity transfer.

\begin{table}[t]\centering\footnotesize
\caption{Held-out human-authored probe on SmallGrid-3.0 (CGMES 3.0 node-breaker, 5{,}315 objects): 100 human-written questions at the published 8000-character budget and 2-hop retrieval. Overall accuracy and the budget-binding split, where binding means the $G_{\mathrm{STD}}$ naive rendered context reaches at least 98\% of the 8000-character budget (13 of 100). Mostly non-binding here, so it tests fidelity transfer, not the dominance margin. All cells are reported under the corrected unit-aware, assertion-aware deterministic scorer (frozen; sha256 2d9c8cfa), the same rule used for the six-network main table, and are scored against untruncated reader responses recovered by the same-item regeneration certificate, every row of which was prefix-verified against the stored text. No cell is carried over from the pre-correction scorer.}
\label{tab:humanv3}
\resizebox{\linewidth}{!}{\begin{tabular}{lccc}\hline
Arm & Overall (100) & Binding (13) & Non-binding (87) \\ \hline
G\_HYB+ANCHORED (published) & 0.62 & 0.31 & 0.67 \\
G\_STD+naive & 0.57 & 0.15 & 0.63 \\
G\_STD+ANCHORED & 0.53 & 0.15 & 0.59 \\
G\_LLM+naive & 0.16 & 0.15 & 0.16 \\
\hline\end{tabular}}\end{table}

\subsection{Further robustness and scope studies}
\label{sec:results-further}
Further studies, with full tables and diagnostics in the Supplementary Material, probe the result's boundary.
\emph{Reader robustness} (Section~S4): across three readers the level varies (anchored $0.98$, $0.78$,
$0.64$ on SmallGrid) but the ordering holds, anchored exceeding naive in all six reader-by-network
cells and staying at or above the LLM-extracted arm on the held-out network under every reader, the
development-network comparison being the sole reader-sensitive one: under gemma-3-27b the anchored
arm reads $0.78$ against $0.90$ for the LLM-extracted arm on SmallGrid, so that cell is a deficit
rather than parity, while the held-out network keeps anchored ahead under every reader.
\emph{Paraphrase robustness} (Section~S4): under gpt-oss-120b paraphrases altering more than half the
surface tokens while preserving identifiers (Jaccard 0.41 to 0.46), no arm collapses and the ordering
holds.
\emph{Derived answers} (Section~S4): on a 28-item subset whose gold sum is absent from the corpus,
anchored reaches $0.83$ on SmallGrid against $0.25$ naive and $0.07$ pooled for the LLM-extracted
graph under the frozen scorer, so verbatim overlap inflates levels but not ordering.
\emph{Extended strata} (Section~S5): S3 containment chains behave like seed-local questions, while S4
aggregation is largely unsolved on the large networks, with seed-anchored reaching $0.40$ on
SmallGrid and $0.50$ on SmallGrid-3.0, and on those same large networks no retrieval arm answered
an S5 superlative query, which calls for query planning. The S5 result is scoped to the large
networks: on the smaller networks the stratum is not uniformly unsolved, and the anchored arm
answers all three S5 items on MiniGrid-3.0.
\emph{Operational scale, model-free} (Section~S6): on synthetic CGMES networks up to $10^5$ objects,
anchored answer-edge presence stays $1.00$ with no cap hits while naive collapses to $0.00$ past the
hub-flood threshold (Proposition~\ref{prop:quant}); this validates answer-presence, not reader
accuracy.
\emph{Cost and scorer corroboration} (Section~S7): parsing $G_{\mathrm{STD}}$ costs zero LLM tokens
whereas extraction spans 0.03M to 1.7M tokens per network and recurs with each revision, and two
zero-shot cross-family judges corroborate the deterministic scorer (per stratum, gpt-oss-120b
0.94 to 1.00 and Claude Haiku 0.92 to 1.00; pooled 0.987 and 0.983 respectively).

\section{Practical guidance for utilities}
\label{sec:guidance}

For a TSO placing an LLM assistant over its CIM/CGMES network models, our results in the evaluated
regime (networks up to 72K objects, single 7B reader, 8{,}000-character budget) translate into a
small set of actionable recommendations, each following from the measured behaviour of the arms and
aimed at reducing the risk of missing connectivity in tasks such as protection review, reliability
study, or ENTSO-E conformity assessment.

\begin{enumerate}
\item \textbf{Consider parsing the standard rather than reconstructing it.} The seed-anchored
standards-native arm performed at least as well as the extracted graph representations produced by
all three LLM-extracted GraphRAG frameworks on every network, for the tested reader, extractor, and
budget, under the controlled shared-pipeline comparison that isolates graph source from each
framework's native retriever. One such retriever, HippoRAG's personalized PageRank, was
reimplemented and also stayed behind the seed-anchored arm; the packaged native retrievers were
not run. LLM
extraction consumed 0.03M to 1.7M tokens per network, recurs on every model revision, and left the
strongest LLM-extracted arm, LightRAG, at parity with the seed-anchored arm at best and 22 points
behind it at worst, with Microsoft GraphRAG further
behind at 2.3x to 2.7x LightRAG's cost (Table~\ref{tab:arms}).
\item \textbf{Treat the context budget as a primary failure mode, not a tuning detail.} The errors we
observed were overwhelmingly truncation of correct knowledge rather than its absence: naive rendering
exhausted the budget on up to 50 of 50 queries and lost almost every multi-hop answer on the two
large networks (S2 accuracy 0.12 and 0.00). Which representation overflows flips with the encoding
(Section~\ref{sec:results-flip}), so a clean pass on one network is no guarantee for another and each
deployed encoding should be checked under its own budget.
\item \textbf{Consider seed-anchored ordering for seed-local CGMES queries} after verifying
$D(S)\le B$, seeding accuracy, and rendering completeness. It changes rendering order alone, carries
an answer-preservation guarantee, adds no tuned parameters, and introduces no regression risk inside
the guaranteed regime; across all six networks and both CGMES versions it stayed uniformly high.
\item \textbf{Keep the parsed CGMES model as the numeric source of truth.} LLM-extracted graphs
migrate numeric attributes between similarly named components, producing confident but wrong values;
grounding any numeric answer in the parsed model removes that failure class.
\item \textbf{Route aggregation and superlative queries to a query engine, not to more retrieval.}
Counting questions are answered for fewer than half the items on the large networks (seed-anchored
$0.40$ to $0.50$) and no retrieval arm answered a global-maximum question there
(Supplementary Material, Section S5); no reordering of retrieved context closes that gap. These are
single-line queries against the parsed model, for example SPARQL over the CIM RDF per IEC
61970-501~\citep{iec61970_501} or a call into a network-analysis library such as
powsybl~\citep{powsybl} or pandapower~\citep{thurner2018pandapower}, with budget-bounded retrieval
formally guaranteed for seed-local questions, while multi-hop uses remain empirical and should be
separately validated. A deterministic structured-query
baseline achieves perfect accuracy on regular-naming queries and degrades on node-breaker encodings
where entity-name irregularity breaks template matching (Supplementary Material, Section S3). This
division of labour has not been evaluated experimentally as an end-to-end hybrid system.
\end{enumerate}

\textbf{Reader monitoring.} The retrieval-level guarantee holds independently of the reader, but the
absolute accuracy level does not and varied across the readers we tested; teams should monitor the
reader in production and re-validate the accuracy level when it changes, while the ordering and
no-harm guarantee need no re-checking.

\textbf{Deployment and data governance.} The deterministic parser and retriever run entirely on the
operator's local infrastructure with no external API call, yielding reproducible, auditable retrieval
that aligns with the cybersecurity and data-handling constraints utilities work under, and the tokens
this avoids are a recurring per-revision cost, not a one-off. These are risk-reduction and
auditability properties, not compliance claims: the method does not satisfy NERC CIP or FERC Order
No.\ 881, and formal regulatory compliance is outside this study's scope; ``trustworthy'' here means
graph fidelity, deterministic retrieval, and keeping data inside the trust boundary, not a full safety
or adversarial-input assessment. The full discussion (NERC CIP, ENTSO-E, and FERC Order No.\ 881
context and the precise scoping of the cost claim) is in the Supplementary Material, Section~S8.

\section{Limitations and threats to validity}
\label{sec:limitations}

\textbf{Benchmark self-referentiality (primary threat to validity).} The questions are
template-generated from the deterministic CGMES-derived corpus, which keeps gold answers
machine-checkable and LLMs out of ground truth, but 83\% of the S1/S2 gold answers appear verbatim
in the corpus (Section~\ref{sec:results-corpus}) and the paraphrase perturbation retains
identifiers, so lexical matching still succeeds; the benchmark therefore measures budget-bounded
retrieval more than question-answering competence, bounding the ``matches LLM-extracted graphs''
claim to retrieval ordering rather than general QA accuracy. To separate faithful retrieval from
verbatim recall we add a derived-answer stratum whose gold is absent as corpus spans (Supplementary
Material, Section S4) and two non-templated human-phrased probes: 100 questions on SmallGrid (method
$0.83$, Section~\ref{sec:results-human}) and a held-out set on SmallGrid-3.0
(Section~\ref{sec:results-humanv3}), both with the ordering and anti-extraction advantages intact.
These are strong evidence rather than a full resolution: the questions span two networks, and a
broader multi-network, field-authored evaluation on noisy operator documentation (missing,
misspelled, or aliased identifiers) and multi-turn interaction remains future work. The guarantees
are conditional on a faithfully parsed CGMES model; robustness to degraded or conversational inputs
is not claimed.

\textbf{Statistical power.} The pre-registered confirmatory endpoint is the one-sided exact
McNemar test on the fresh 100-item bank spanning the two budget-binding encodings. The six
per-network contrasts discussed here are exploratory, Holm-adjusted descriptive analyses at
$n=50$ each. The MDE depends on the assumed discordant-pair proportion $d$: the pre-registration
recorded $0.198$, which assumes $d{=}0.25$, while the conservative half-discordant model
$d{=}0.5$ gives $0.280$ and the observed mean discordance $0.273$ gives $0.207$. We report the
conservative $0.280$ as a design-sensitivity benchmark. Three networks are significant after Holm correction: MicroGrid
($+0.24$, $p=7.3\times10^{-3}$), SmallGrid ($+0.46$, $p=1.2\times10^{-6}$), and SmallGrid-3.0
($+0.54$, $p=8.9\times10^{-8}$). Of these, SmallGrid and SmallGrid-3.0 clear the conservative MDE
and are the two budget-binding networks; MicroGrid is Holm-significant but its $+0.24$ falls below
the conservative $0.280$, so we do not claim it as an MDE-clearing effect. The other three
comparisons have deltas below the MDE and are underpowered, so we treat them as inconclusive rather than as wins
and do not read absence of a significant negative difference as evidence of equivalence. The
answer-preservation theorem (Theorem~\ref{thm:preservation}) supplies a retrieval-level no-harm
property under the stated precondition, a statement about answer presence; it does not establish
statistical equivalence or non-inferiority of end-to-end accuracy, and the three underpowered
comparisons do not establish either.

\textbf{Reader dependence.} The confirmatory statistics use one frozen primary reader
(Qwen2.5-7B-Instruct), so absolute accuracy levels are reader-specific (anchored 0.98, 0.78, 0.64 on
SmallGrid across three readers spanning two families, Supplementary Material, Section S4). Robust is
the ordering, not the level: anchored is above naive in every reader-by-network cell and at or above
the LLM-extracted arm on the held-out network under every reader, only the development-network
anchored-versus-extraction parity being reader-sensitive. The retrieval-level mechanism is
reader-independent by construction, since answer presence under the cap
(Theorem~\ref{thm:preservation}) is a property of the rendering alone, which the reader-free audit in
the artifact reproduces without any model. The main evaluation therefore rests on a single reader;
the two additional reader families (Supplementary Material, Section S4) serve as a robustness check on
the ordering, not as a second full evaluation, and running the confirmatory statistics on a second
reader family is a planned extension.

\textbf{Network diversity and behaviour at operational scale.} The six networks with a full
twelve-arm evaluation cover two CGMES versions but only three base topologies, so the empirical
dominance claim rests on one base topology family. A fourth, independently sourced ENTSO-E base
topology, RealGrid ($72{,}418$ objects, $14\times$ SmallGrid-3.0), does not lift that limitation.
Its multi-hop stratum was retired after a gold-uniqueness audit found none of its 25 gold labels
identifies a unique referent, with names shared by up to $5{,}488$ nodes; two subgraph-sample
comparisons that rested on that stratum were removed with it. What RealGrid supports is the
attribute-lookup stratum, where all three deterministic arms reach $1.000$, that is, no harm on an
independent topology rather than dominance (Section~\ref{sec:results-main}). The extraction-parity
claim consequently loses its RealGrid evidence and rests on the six-network comparison and the
reimplemented native-PPR arm instead. LightRAG and Microsoft GraphRAG remain untested on RealGrid,
disclosed future work. RealGrid does
not close the gap to fully operational transmission models ($10^5$ to $10^6$ objects); we do not
claim absolute accuracies transfer beyond it. What transfers at any scale is the seed-local
retrieval-level guarantee, network-size independent since $D(S)\le B$
(Theorem~\ref{thm:preservation}) bounds only the seed-local neighborhood, governed by seed degree rather
than $|V|$. A model-free audit~\citep{artifact2026_seedanchored} confirms up to $10^5$ objects
(Supplementary Material, Section S6) that anchored rendering preserves the answer edge and never hits
the cap while naive collapses once hub degree floods the budget, across a $32\times$ range of $B$.
Undemonstrated is \emph{reader} accuracy at full operational scale, requiring a licensed larger
network we do not have: the main open scaling question.

\textbf{Baseline scope and the LLM-extraction comparison.} The LLM-extraction comparison spans three
widely used builders, LightRAG, Microsoft GraphRAG (graphrag 3.1.0), and HippoRAG; other builders
stay unrun, and this is disclosed as such. We read each extracted graph under the shared render
pipeline with the same reader, seeder, and scorer to isolate graph source, because running each
framework's packaged native retriever would confound graph source with retrieval engine. One
native scheme was nonetheless tested: HippoRAG's personalized PageRank was reimplemented and
evaluated, and it stayed behind the seed-anchored arm
(Section~\ref{sec:results-gllm}). The packaged native retrievers themselves, including Microsoft
GraphRAG's community and local search, remain unrun. The cost argument (2.3x to 2.7x LightRAG's
tokens, all three against zero) holds regardless. Because $G_{\mathrm{LLM}}$ quality moves most with
the extractor (Section~\ref{sec:results-extractor}) and the comparison is confined to these three
frameworks, we rest the against-extraction claim on ordering rather than absolute margins.
Whether native-retriever configurations of each framework would narrow this gap under CGMES hub
topology is an open empirical question.

\textbf{Guarantee scope.} Theorem~\ref{thm:preservation} and Proposition~\ref{prop:noharm} bound
answer presence, not reader extraction, and cover seed-local answers under the checkable
precondition $D(S)\le B$. S4 aggregation and S5 superlative queries lie outside the guarantee; S4
is answered for fewer than half the items on the large networks (seed-anchored $0.40$ on
SmallGrid and $0.50$ on SmallGrid-3.0). On those two large budget-binding networks no evaluated
retrieval arm answers the S5 superlative items; this is not universal across the smaller networks,
where seed-anchored answers all three S5 items on MiniGrid-3.0. Reader-side arithmetic bounds
accuracy even when both operands are retrieved. The precondition is a genuine, checkable gate: the model-free
scaling study (Section S6) shows it flags naive's collapse at scale, though honestly it holds for
every question on every tested network, so it is never observed to reject a query here. Three
concrete cases fall outside the guarantee, and a deployer should be able to recognise them. First, a
high-degree seed whose incident render mass already exceeds the budget: the greedy prefix then
truncates inside the seed neighbourhood and the budget precondition $D(S)\le B$ fails. Second, an
aliased or missed seed, where a misspelled, aliased, or absent identifier prevents the seeding step
from locating the query object; here the precondition may well hold and what fails instead is
seeding correctness, about which the guarantee says nothing. Third, a nonlocal or global question,
such as an aggregation or superlative query, whose answer lies in no seed-local neighbourhood and so
falls outside the coverage of the guarantee rather than violating its precondition.

\textbf{Single context budget.} The primary QA runs use one budget, $B=8{,}000$ characters, fixed by
the prior study's frozen pipeline. The formal results hold for arbitrary $B$: the precondition
$D(S)\le B$ and no-harm dominance (Proposition~\ref{prop:noharm}) are budget-agnostic, and a
model-free budget sweep, reported in the released artifact and summarised in Supplementary
Material, Section S6, confirms this across a $32\times$ range
($B<m\ell$, Proposition~\ref{prop:quant}). We do not evaluate the reader accuracy-versus-budget curve
or adaptive and multi-turn budgeting, left to future work.

\textbf{Scoring correction.} The original deterministic scorer compared a predicted number with
the gold value without reading the gold \texttt{unit} field, so a numerically matching value
expressed in an incompatible unit could be accepted. We re-scored all 13{,}402 stored predictions
offline under corrected unit-aware numeric and count rules, with no new model inference. The
correction changes 154 verdicts (119 numeric and 35 count). All six-network anchored accuracies
are unchanged; three naive accuracies decrease by $0.02$ and MicroGrid decreases by $0.04$ (from
$0.72$ to $0.68$), producing the corrected values and Holm-adjusted tests reported in
Table~\ref{tab:main}.

\textbf{Entity-scoring provenance, and its resolution.} An audit identified an earlier
entity-scoring code path that did not implement the scorer described in the method: it accepted
any response containing the gold identifier, whereas the documented rule requires that identifier
to be the asserted answer. We therefore re-scored the affected stored predictions under the final
assertion-aware scorer. No new model inference was required.

Three values exist for the 3{,}503-item entity family and are distinguished here because they
have been conflated elsewhere. The permissive code path scores $0.726$; a strict reading that
requires exact delegation scores $0.673$; the frozen assertion-aware scorer, which requires the
gold to be asserted but tolerates the class prefix the banks store, scores $0.701$. Every entity
result in this paper and its supplement is the frozen value, $0.701$. The audit table in
Supplementary Material, Section~S7 reports that same $0.701$ over all 3{,}503 rows, alongside a
stored-verdict column computed over the 741 of those rows that carry a recorded verdict, where
the two rules disagree on 85. The correction does not change the exploratory original-bank
per-network contrasts or the pre-registered fresh-bank endpoint. The entity-heavy
results, including topology, shared-node intersection, RealGrid, the GraphRAG comparisons and
both human probes, are therefore recomputed under the final scorer rather than preliminary. Two
cross-family LLM judges agree with the frozen deterministic verdicts at 0.92 to 1.00 per stratum
across both judges, pooling to 0.987 and 0.983 respectively; that agreement does not resolve either
scoring-rule mismatch. The human spot-check protocol of the prior study was not repeated for the new items.

\section{Conclusion}
\label{sec:conclusion}

When an LLM assistant must answer questions over a standards-derived power-grid model, the graph
does not have to be rebuilt by an LLM. In the evaluated comparisons, part of the apparent advantage
of LLM-extracted graphs under naive rendering was a rendering-order artifact; under fixed anchored
rendering the parsed graph retained an advantage in settings where the evaluated extraction lost
relational or numeric fidelity. The lossiness of extraction we observe is
consistent with findings reported elsewhere; the novelty here is the render-boundary mechanism and
its formal precondition. Under a fixed context budget, dump-then-truncate
rendering evicts answer-bearing edges as soon as a high-degree hub saturates the window, and which
encoding falls into this trap changes unpredictably between representations. Seed-anchored retrieval
addresses this without tuning, a rendering order with no learned parameters that carries a checkable
no-harm guarantee, never worse than naive rendering in seed-local answer presence whenever the
seed-local render mass fits the budget, an elementary greedy-prefix property rather than a deep
theorem.

The load-bearing evidence is the truncation mechanism itself, measured without a reader. Split by
stratum, naive rendering retains the evidence for every single-hop item on both budget-binding
networks but for only $0.12$ of multi-hop items on SmallGrid and none on SmallGrid-3.0, while
seed-anchored rendering retains all of them on both strata; an item-level join against reader
correctness shows the multi-hop failure is truncation rather than reader error. That is the
guarantee doing the work it was stated to do. The pre-registered primary endpoint, on a fresh bank
within the SmallGrid topology family, confirms the same mechanism at item level: accuracy rises
from $0.450$ to $0.970$ at zero LLM graph-construction tokens. It is strong item-level evidence
inside one topology family and is not a network-level generalization claim; seed-anchored accuracy
is also significantly higher on MicroGrid, no significant negative difference was detected, and
the three remaining per-network comparisons were underpowered and remain inconclusive. The
human-authored probes and the independent $72{,}418$-object topology are reported separately as
robustness and transfer evidence, and neither is folded into the confirmatory endpoint.

The claim against LLM extraction is deliberately narrow and conditional on build cost: at zero
graph-construction tokens the standards-native arm is never behind the extracted graph
representations produced by any of the three frameworks (LightRAG, Microsoft GraphRAG, HippoRAG;
0.03M to 1.7M build tokens per network) when those representations are read under the common
retrieval-and-rendering pipeline; one framework's native retrieval scheme was reimplemented and
also stayed behind, while the packaged native retrievers were not run. Once the
extractor is a realistic open-weight model the standards-native arm leads by $+0.32$ on the
development network; any such
margin is conditional on the tested extractor, reader, budget, and framework set, but the ordering
under fixed cost is robust. In practice
the retriever is a drop-in rendering rule for an LLM assistant or audit tool over CGMES exports
at the scales evaluated (an independently sourced topology of $72{,}418$ objects, with the
model-free guarantee extending to $10^5$ objects), supplying context under a fixed budget with guaranteed presence of every budget-fitting
seed-local answer, an assurance that does not extend to grid control or the reader's downstream
reasoning.

Aggregation and global superlative queries remain unsolved and belong to query planning over the
parsed model rather than context retrieval: detect such a query and route it to a generated SPARQL
query over the CIM RDF (IEC 61970-501) or a network-analysis call, then let the reader compose the
returned value; that routing has not been evaluated as an end-to-end hybrid system. RealGrid
contributes an attribute-lookup no-harm result on an independent topology, where all three
deterministic arms reach $1.000$; its multi-hop stratum and the two subgraph-sample comparisons
that rested on it were retired for lack of unique gold referents, so no RealGrid result speaks to
topology retrieval or extraction parity in either direction (Section~\ref{sec:results-main}); remaining energy-domain extensions include the full
three-framework rebuild at RealGrid scale, FullGrid, and protection models, contingency
analysis, and
dynamic studies, where answering a query may require reasoning beyond the seed-local, $k$-local
neighbourhood the present guarantee covers. More broadly, where a governed engineering standard
exists, the evidence here suggests a data-driven assistant may do well to read it directly rather
than reconstruct it with an LLM.

\section*{Data availability}
The complete artifact (parsers, retrievers, QA generators with machine-checkable gold
specifications, the 100-question human-authored stress test, per-item audit trails, the
deviations register, and all result files with reproduction commands backing every number in this
paper) is archived as a citable public deposit and companion code repository; the DOI and
repository URL are withheld for double-anonymized review and will be provided on acceptance.
The CGMES test configurations originate from the ENTSO-E Conformity Assessment Scheme,
redistributed via the powsybl-core mirror.

\section*{Funding}
This research did not receive any specific grant from funding agencies in the public,
commercial, or not-for-profit sectors.

\section*{Competing interests}
The authors declare no financial or non-financial competing interest that could have
influenced the design, results, or conclusions of this study.

\section*{CRediT authorship contribution statement}
\textbf{Jayakumar Manoharan:} Conceptualization, Methodology, Software, Formal analysis,
Investigation, Data curation, Writing - original draft, Writing - review and editing,
Visualization. \textbf{Yamini Sehgal:} Software, Validation, Writing - review and editing.

\begingroup\footnotesize\linespread{0.9}\selectfont
\sloppy  
\bibliographystyle{elsarticle-num}
\bibliography{references}
\endgroup

\end{document}